\documentclass[10pt,journal,compsoc]{IEEEtran}%
\IEEEoverridecommandlockouts
\usepackage{ragged2e}
\usepackage{enumitem}
\usepackage{mathrsfs}
\usepackage{soul}
\usepackage{cite}
\usepackage{amsmath,amssymb,amsfonts}
\usepackage{colortbl}
\usepackage{hhline}
\usepackage{graphicx} 
\usepackage{subfig}
\usepackage{textcomp}
\usepackage{url}
\usepackage{stfloats}
\usepackage{mathtools}
\usepackage{cuted}
\usepackage{array}  
\usepackage {threeparttable}
\usepackage{tablefootnote}
\usepackage{multirow}
\usepackage{colortbl}
\def\BibTeX{{\rm B\kern-.05em{\sc i\kern-.025em b}\kern-.08em
    T\kern-.1667em\lower.7ex\hbox{E}\kern-.125emX}}

\usepackage{stackengine}
\usepackage{amsmath}
\usepackage{algorithm}
\usepackage{algpseudocode}
\usepackage{enumerate}
\usepackage{booktabs}
\usepackage{subfig}
\usepackage{multicol}
\usepackage{pgfplots}
\usepackage{pgfplotstable}
\definecolor{purple}{RGB}{0,0,0}
\definecolor{blue}{RGB}{0,0,0}
\usepackage{setspace}
\usepackage{color}
\usepackage{tabularray}
\definecolor{Iron}{rgb}{0.811,0.815,0.815}
\makeatletter
\def\BState{\State\hskip-\ALG@thistlm}
\makeatother

\algnewcommand\algorithmicforeach{\textbf{for each}}
\algdef{S}[FOR]{ForEach}[1]{\algorithmicforeach\ #1\ \algorithmicdo}
 \usepackage{amsthm}

\newtheorem{remark}{Remark}
\newtheorem{proposition}{Proposition}
\begin{document}

\title{Efficient Multi-user Offloading of Personalized Diffusion Models: A DRL-Convex Hybrid Solution} 
\author{Wanting Yang, Zehui Xiong,~\IEEEmembership{Senior Member, IEEE}, Song Guo,~\IEEEmembership{Fellow, IEEE},  \\ Shiwen Mao,~\IEEEmembership{Fellow, IEEE},  Dong In Kim,~\IEEEmembership{Fellow, IEEE},   Merouane Debbah~\IEEEmembership{Fellow, IEEE}
\thanks{Manuscript received ...\\
Wanting Yang and Zehui Xiong are with the Pillar
of Information Systems Technology and Design, Singapore University of
Technology and Design, Singapore (e-mail: wanting\_yang@sutd.edu.sg; zehui\_xiong@sutd.edu.sg);\\
Song Guo is with the Department of Computer Science and Engineering,
The Hong Kong University of Science and Technology, Hong Kong, SAR,
China (e-mail: songguo@cse.ust.hk); \\Shiwen Mao is with the
Department of Electrical and Computer Engineering, Auburn University,
Auburn 36830, USA (e-mail: smao@ieee.org); \\Dong In Kim is with the Department of Electrical and Computer Engineering, Sungkyunkwan University, Suwon 16419, South Korea (e-mail:
dongin@skku.edu);  \\Merouane Debbah is with KU 6G Research Center, Khalifa University of Science and Technology, P O Box 127788, Abu Dhabi, UAE (email: merouane.debbah@ku.ac.ae) and also with CentraleSupelec, University Paris-Saclay, 91192 Gif-sur-Yvette, France.}}


\IEEEtitleabstractindextext{
\begin{justify}
\begin{abstract}

\textcolor{blue}{Generative diffusion models like Stable Diffusion are at the forefront of the thriving field of generative models today, celebrated for their robust training methodologies and high-quality photorealistic generation capabilities. 
These models excel in producing rich content, establishing them as essential tools in the industry. Building on this foundation, the field has seen the rise of \textit{personalized content synthesis} as a particularly exciting application.}
However, the large model sizes and iterative nature of inference make it difficult to deploy personalized diffusion models broadly on local devices with \textcolor{purple}{heterogeneous} computational power.  \textcolor{purple}{To address this}, we propose a novel
framework for efficient multi-user offloading of personalized diffusion models\textcolor{purple}{. This framework accommodates a variable number of users, each with different computational capabilities, and adapts to the fluctuating computational resources available on edge servers}. To enhance computational efficiency and \textcolor{purple}{alleviate} the storage burden on edge servers, we  propose a tailored multi-user hybrid inference approach\textcolor{purple}{. This method splits the inference process for each user into two phases, with an optimizable split point}. \textcolor{purple}{Initially, a cluster-wide model processes low-level semantic information for each user's prompt using batching techniques. Subsequently, users employ their personalized models to refine these details during the later phase of inference.} 
Given the constraints on edge server computational resources and users' preferences for low latency and high accuracy, we model the joint optimization of each user's offloading request handling and split point as an extension of the Generalized Quadratic Assignment Problem (GQAP). Our objective is to maximize a comprehensive metric that \textcolor{purple}{balances  both latency and accuracy across all users}.  To solve this NP-hard problem, we transform the GQAP into an adaptive decision sequence, model it as a Markov decision process, and develop a hybrid solution combining deep reinforcement learning with convex optimization techniques.  Simulation results validate the effectiveness of our framework, demonstrating superior optimality and low complexity compared to traditional methods. \textcolor{purple}{All related code, datasets, and fine-tuned models are available at \url{https://github.com/wty2011jl/E-MOPDM}}.

\end{abstract}
\end{justify}
}

\maketitle

\IEEEdisplaynontitleabstractindextext

\vspace{0cm}
\section{Introduction}
\label{sec: Intro}
 The remarkable progress of generative models has ushered in a new era of artificial intelligence technology. Notable advancements include ChatGPT in natural language processing, demonstrating capabilities on par with human experts in a wide range of applications. Similarly, Stable Diffusion Models (SDMs)~\cite{rombach2022high}, developed based on Denoising Diffusion Probabilistic Models~\cite{ho2020denoising}, have \textcolor{purple}{gained} comparable prominence in computer vision. The impressive generative capabilities and the training robustness of SDMs have established them as a foundational tool in numerous applications, such as \textit{super-resolution}, \textit{image editing}, \textit{text-to-image} and  \textit{inpainting}, etc.~\cite{croitoru2023diffusion}.
 \textcolor{purple}{They are also pivotal in empowering industries such as telecommunications~\cite{yang2024rethinking} and medical~\cite{kazerouni2023diffusion} sectors.}
 Among these advanced applications, personalized content synthesis (PCS) stands out as \textcolor{purple}{a particularly important and eagerly anticipated branch} of research. 

The customizable features of PCS \textcolor{purple}{enable users to generate images tailored to objects, styles, faces, and other high-level semantic details. For example, this capability allows users to place themselves in virtual scenes like historic landmarks or future cities, or recreate rare pet behaviors, such as a cat leaping gracefully. As a result, PCS has garnered significant public interest, prompting a surge in research and development efforts within this area.}  Since the release of \textit{DreamBooth} and \textit{Textual Inversion} in August 2022, \textcolor{purple}{the field has rapidly expanded, with over 100 methods developed in a short time.~\cite{zhang2024survey}}. 
\textcolor{purple}{However, especially within the PCS field, rising privacy concerns have become particularly pronounced due to the personal nature of the content being generated. Additionally, the increasing number of users has put excessive pressure on cloud resources. Consequently, there is a strong demand for on-device deployment of SDMs to address these issues. } Yet, these models often come with large model sizes and \textcolor{purple}{substantial computational requirements}. In this sense, the slow and iterative sampling process during inference limits \textcolor{purple}{the widespread deployment of personalized SDMs on local devices, as they must adapt to the diverse computational capabilities of user devices.}


Recently, research aimed at \textcolor{purple}{enhancing inference speed in diffusion models has been organized into three main categories}. The first focuses on \textit{sampling step reduction}, which is best known for its implementation in Denoising diffusion implicit models (DDIMs)~\cite{song2020denoising}.  
The second approach is based on \textit{model compression}~\cite{yang2023diffusion, fang2023structural,li2023qdiffusion}. 
The third category achieves \textit{hybrid inference} through edge offloading~\cite{liu2023oms,10172151}. \textcolor{purple}{(More details of these techniques are discussed in Section~\ref{sec:2}.)}
Among these approaches, the DDIM has been integrated into SDMs. However, generating high-quality and detail-rich images still requires hundreds of denoising steps. 
In recent work~\cite{yan2024hybrid}, \textit{ByteDance Inc.} combines the aforementioned model compression and both hybrid inference methods for SDMs. The inference process begins with a smaller model on the user side, while the latter denoising steps, which require additional detail, are performed by a large-scale model on the cloud or edge servers. The simulation results confirm the promising gains brought by this \textit{composite hybrid inference} approach. \textcolor{purple}{However, unlike shared foundation models, personalized diffusion models are unique to each user. Extending this to a practical multi-user scenario, replicating each personalized model on the edge server for independent inference would place immense storage and computational burdens on the servers used for offloading. This approach deviates from the trend toward on-device applications and could potentially exceed the capabilities of even cloud infrastructure.}

\textcolor{purple}{{ To overcome this dilemma and facilitate the widespread adoption of personalized SDMs, this work focus on a balanced solution for multi-user personalized SDMs.}}
\textcolor{purple}{Our goal extends beyond optimizing specific inference algorithms. Instead, we intend to seamlessly integrate inference algorithm optimizations with real-world resource management, relying on the typical edge-device collaborative network framework.}
Specifically, we propose a novel framework for efficient multi-user offloading of personalized diffusion models. This framework \textcolor{purple}{addresses scenarios with} a variable number of users, diverse user computational capabilities, and \textcolor{purple}{fluctuating edge server resources}. \textcolor{purple}{To minimize redundancy and alleviate storage burdens, we employ a cluster-wide model that captures common personalized features of the cluster's users through fine-tuning, serving as a shared offloading model. Furthermore, we enhance computational efficiency by enabling batch processing of offloading tasks for multiple users. Nevertheless}, these improvements also present two new challenges.
\begin{itemize}[left=0pt]
    \item The hybrid inference approach, which combines the offloaded inference by the cluster-wide model with local inference by the personalized model. While offloading reduces inference latency, it also  introduces a degree of personalization loss. Thus, a trade-off between latency reduction and accuracy of personalized feature representation must be carefully considered when optimizing offloading steps.
    \item \textcolor{purple}{Existing approaches often overlook the server's parallel inference capabilities and assume that total latency is proportional to the number of tasks. Unlike these methods, the batching technique presents a unique challenge. As batch size increases, it leads to longer inference latency for each denoising step, thereby exacerbating the interdependence between individual user offloading decisions and overall system performance.}
\end{itemize}
To address the aforementioned challenges, we innovatively develop a comprehensive and precise optimization metric. Using this metric as a guide, we formulate the efficient multi-user offloading for personalized diffusion models as an extension of the Generalized Quadratic Assignment Problem (GQAP). Additionally, we design a low-complexity solution specifically for this extended version of GQAP. The key contributions of our approach are summarized as follows:

\begin{itemize}[left=0pt]
    \item We propose an efficient multi-user hybrid inference manner. In this framework, a cluster-wide SDMs, which is trained to capture common features of personalized demands among users within the cluster, is located at the edge server. Meanwhile, the personalized model of each user is deployed at the local individually. 
The inference process for each user is split into two phases with an \textit{optimizable split point}. For users participating in offloading, the initial phase of inference is performed on the cluster-wide model using batching techniques. This approach synchronously generates low-level semantic and contour information tailored to each user’s prompt. The intermediate results are subsequently returned to each user, who then applies a personalized SDM to add finer details in the later stages of the generation process.
\item We propose a tailored metric that balances latency and accuracy of personalized feature, interconnected through a weighting parameter. This parameter captures each user's preference between the two factors and can be customized individually to reflect their specific requirements. Based on rigorous mathematical analysis, we provide the range for this parameter along with insights into the corresponding effects. To quantify accuracy in PCS, we define a metric called Personalized Accuracy Index (PAI), which jointly considers the traditional accuracy, measured by the CLIP, and the variation in final generated images for different users given the same prompt, assessed  by LPIPS. We train a cluster-wide model and multiple personalized models using  DreamBooth and establish a unified mathematical model describing the relationship between PAI and the split point based on empirical simulations.
\item We formulate the joint optimization of offloading request handling and each user's individual split point as an extended GQAP, aiming to maximize overall performance using the tailored metric as a guiding principle. To meet the real-time decision-making requirements, we are, to the best of our knowledge, the first to apply Deep Reinforcement Learning (DRL) to solve GQAP, proposing a DRL-convex hybrid solution. Specifically, we introduce a novel paradigm that transforms GQAP into a sequential decision-making process and maps it into a Markov Decision Process (MDP) model. Meanwhile, we determine the optimal split point for each user based on convex optimization theory and integrate it into the environment of the MDP. We compare this algorithm with classical algorithms for solving GQAP, such as Branch \& Bound and heuristic algorithms, demonstrating its optimality and low  complexity, $O(n)$, regarding the number of users.
\end{itemize}

The rest of this paper is organized as follows. In Section~\ref{sec:2},
we review research in hybrid inference for SDMs and the existing solutions relevant to GQAP. Then, in Section~\ref{sec: systemmodel}, we introduce
the system overview, and  relevant models.  Then, we present the problem formulation in Section~\ref{sec: formulation} and the   DRL-convex hybrid solution is described in Section~\ref{sec: DRL-COnvex}.  An experimental evaluation is
presented in Section~\ref{sec:simulation}, followed by the conclusion and future directions in Section~\ref{sec:conclusion}.



\section{Related works}
\label{sec:2}
\vspace{-0.5cm}

\textcolor{purple}{\subsection{Inference Latency Reduction in Diffusion Models}
As discussed in Section~\ref{sec: Intro}, research aimed at overcoming these challenges can be organized into three main directions. The first direction, Sampling Step Reduction, is exemplified by Denoising Diffusion Implicit Models (DDIMs)~\cite{song2020denoising}, which reparameterize the diffusion process into a non-Markovian chain. This enables sampling in 50 - 200 steps, or even fewer, instead of the typical 1,000, but at the cost of reduced high-frequency detail generation and increased output variability. Efforts to balance speed with fidelity include adaptive step scheduling and noise-aware sampling, but challenges remain.
The second direction, Model Compression, includes architectural redesign, pruning, and quantization to streamline diffusion architectures while preserving semantic coherence. Architectural redesigns, such as replacing U-Net backbones with vision transformers in U-ViT~\cite{bao2023all}, achieve significant parameter reductions while maintaining competitive FID scores. Structured pruning strategies~\cite{fang2023structural} remove redundant elements to reduce computational demands, and quantization techniques demonstrated by Q-Diffusion~\cite{li2023qdiffusion} employ precision adjustments to maintain performance. Despite advancements, compressed models struggle with complex compositions and rare concepts, often requiring task-specific adjustments to avoid semantic drift.
The third category, \textit{hybrid inference}, includes (1) combining models of different sizes to offset quality degradation from exclusive reliance on compressed models~\cite{liu2023oms}, and (2) leveraging edge server's greater computing power to offload part of the denoising steps, enabling collaborative inference and reducing overall latency~\cite{10172151}.  A detailed discussion on hybrid inference will follow in Section~\ref{sec:2-1}.}
\begin{table}[]
\footnotesize
 \centering
 \caption{Summary of Main Notations}
\renewcommand{\arraystretch}{1.1}
\begin{tabular}{m{0.8cm}|m{7.3cm}}
\hline
\multicolumn{2}{c}{\textbf{System Setting-Related Symbols}}                                                                                                          \\ \hline \hline
$G$                         & The number of available GPUs at edge \\ \hline
{$\mathcal{I}$}            & {The set of user sending requests}                                        \\ \hline
{$B_{\rm max}$}      & {The maximum batchsize supported by the edge server}                        \\ \hline
{$N$}     & {The total denoising steps during inference}                       \\ \hline
{$\hat N$}                    & {The maximum number of offloaded denoising steps}                                              \\ \hline
{$W_{\rm max}$ } & {The total bandwidth shared by the granted users}                                                           \\ \hline
{$s^{\rm p}$}                & {The data size of prompt}                                                   \\ \hline
{$s^{\rm m}$}               & {The data size of intermediate data}                                      \\ \hline
{$\eta $}                       & {The spectral  efficiency over wireless link} \\ \hline \hline 
\multicolumn{2}{c}{\textbf{Metric and Decision-Related Symbols}}                             \\\hline \hline 
$\ell^{\rm R}$                                & The round-trip time (waiting latency) \\ \hline
$\ell^{\rm T}$                              & The transmission latency \\ \hline 
$\ell^{\rm E}$                              & The computing latency at edge \\ \hline 
$\ell^{\rm L}$                              & The transmission latency at local \\ \hline 
$L$                              & The end-to-end latency \\ \hline 
$n^*$                                 & The optimal split point  \\\hline 
$\mathbf{x}$                                    & The results of request handling      \\ \hline 
$b_{\rm g}$                                    & The number of granted user for offloading      \\ \hline
$\alpha_i$                                    & The customized emphasis parameter between $L$ and PAI     \\ \hline
$\mathscr{L}_e\left(\cdot\right) $                                    & The function of  inference latency per denoising step at edge on $b_{\rm g}$ and $G$    \\ \hline
$\mathscr{F}\left( \cdot\right) $                                    & The function of  personalized accuracy index metric on $n^*$    \\ \hline
\end{tabular}
\label{tbl: symbols}
\vspace{-0.5cm}
\end{table}

\subsection{Hybrid Inference for Diffusion Models}
\label{sec:2-1}
Since the remarkable success of diffusion models in generative tasks, research on latency optimization for inference processes has gained momentum. As discussed in Section~\ref{sec: Intro}, this research initially focused on reducing denoising steps and compressing the model. However, both approaches inevitably lead to some degree of performance degradation; the greater the reduction in inference latency, the more noticeable the performance loss. Beyond these two approaches, hybrid inference has been proposed, offering flexible combinations with the above methods.

In 2022, \textit{NVIDIA Corporation} first proposes the eDiff-I framework, where the inference process is divided into multiple stages. Each stage corresponds to a distinct model, focusing on different generative functions to enhance overall inference performance~\cite{balaji2022ediff}. Then, the authors in~\cite{pmlr-v202-liu23ab} design OMS-DPM, a predictor-based search algorithm, to determine the optimal model schedule given a specific generation time budget and a set of pre-trained models. Recently, in~\cite{yang2024denoising}, the authors present a new framework called DDSM, which employs a spectrum of neural networks whose sizes are adapted according to the importance of each generative step. 
In addition to the model's size and functionality, the authors in~\cite{du2023exploring} expand the concept of hybrid inference to incorporate computational heterogeneity within the inference process. Here, the inference process is split, with one part completed by a more powerful edge server and the latter part handled locally by the user device, connected via a wireless transmission link to reduce inference latency. Additionally, they propose an interesting idea: users with similar prompts can share the initial portion of the inference process on the edge server, reducing the edge server’s computational load. However, this inevitably comes at the cost of lower image generation quality. To address this, the authors in~\cite{xie2024gai} further optimize  the split point under this setting. Although the framework’s effectiveness is applicable when prompts from different users are highly similar, in practice, user prompts are often diverse. Therefore, this framework’s use cases are quite limited.

In addition, ByteDance further extends hybrid inference to simultaneously incorporate both model heterogeneity and computational heterogeneity, as introduced in~\cite{yan2024hybrid}. However, they only consider the single-user inference scenario and do not address the inevitable challenges in real-world applications, where cloud server resources available to users are limited. Consequently, they do not explore the necessary handling of offloading requests or the optimization of the split point based on personalized user requirements under such conditions. Overall, optimization of diffusion model inference for real-world applications remains in a very early stage.
\begin{figure*}[h]
    \centering
    \includegraphics[width=0.9\linewidth]{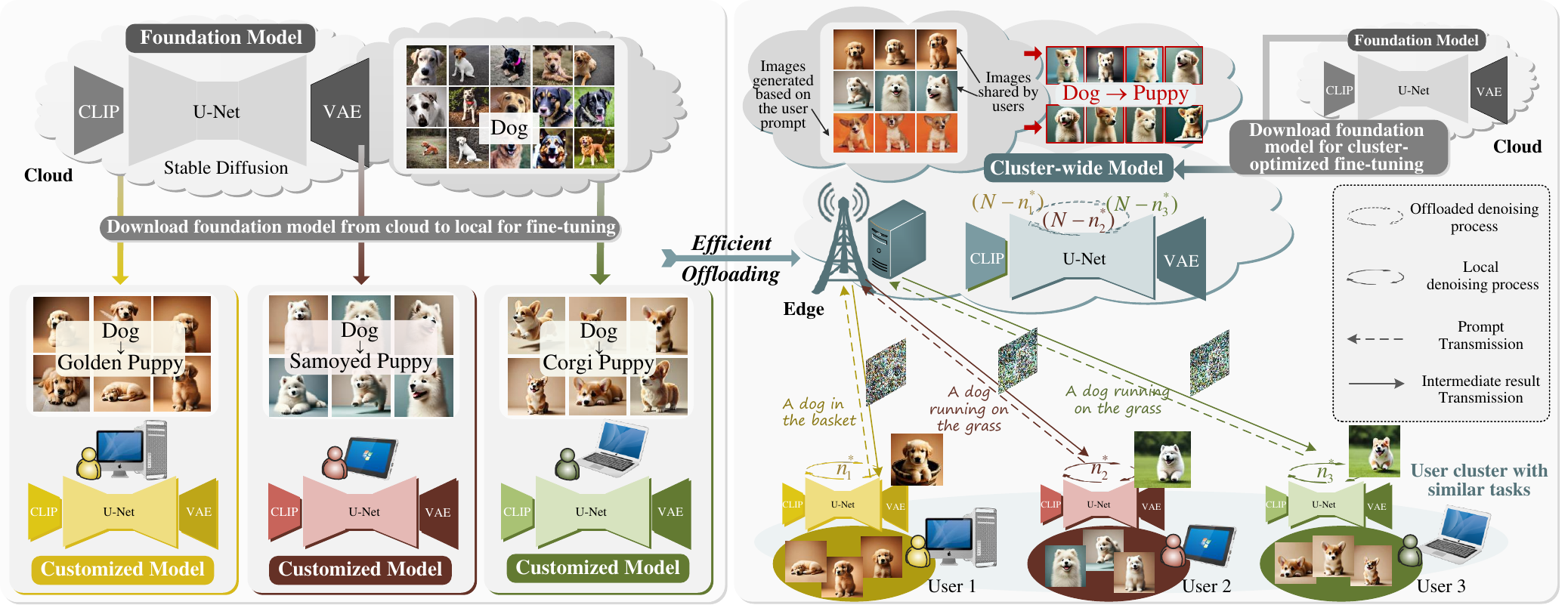}
    \caption{Illustration of system overview and pre-deployment approach.}
    \label{fig:systemmodel}
    \vspace{-0.5cm}
\end{figure*}
\subsection{Generalized Quadratic Assignment Problem}
\label{sec:2-3}
Combinatorial optimization problems (COPs) are widely present in network resource optimization, including the \textit{traveling salesman problem}~\cite{10598354}, \textit{knapsack problem}~\cite{liu2023resource}, and \textit{multi-armed bandit problem}~\cite{10444936}, etc. Among them, quadratic assignment problem (QAP) is considered as one of the most difficult to solve in combinatorial optimization~\cite{silva2021quadratic}. In addition, various variants of this problem, such as \textit{Quadratic Bottleneck Assignment Problem}, Biquadratic QAP, quadratic semi-assignment problem  and the relatively new GQAP, have continuously emerged, attracting significant attention over the past decades. A common approach is to linearize the quadratic terms using linearization techniques, enabling the problems to be solved with general-purpose mixed integer linear programming solvers. Two widely used techniques are standard linearization~\cite{glover1974converting} and the reformulation-linearization technique (RLT)~\cite{sherali1998reformulation}. Meanwhile, the authors in~\cite{pessoa2010algorithms} further integrate Lagrangean relaxation procedure over RLT, then combine with the classic Branch \& Bound scheme. Nevertheless,
the quadratic form of its objective function makes it challenging to find even an approximate optimal solution within limited time, especially when the problem size is slightly larger. To this end, the authors in~\cite{silva2021quadratic} propose a heuristic algorithm, a parallel iterated tabu search algorithm. Although it outperforms traditional integer linear programming solvers in terms of performance, its high computational complexity makes it challenging to apply in real-time optimization scenarios.

Meanwhile, the success of DRL in solving increasingly complex problems over the past decade has also inspired a wave of research on using DRL to tackle COPs with low complexity. Notable achievements have been made in solving problems such as the TSP, Vehicle Routing Problem, and Maximum Cut~\cite{khalil2017learning,ma2019combinatorial,cappart2021combining}. The basic idea is to transform the optimization of a high-dimensional variable into a decision sequence, thus formulating it as an MDP model based on the DRL paradigm. However, unlike ordering problems such as the TSP, the QAP problem \textcolor{purple}{hardly permits such a transformation}. Thus using DRL to solve QAP has received relatively little attention. To our knowledge, only two studies have explored this approach. In~\cite{pashazadeh2021difficulty}, the authors discuss the challenges of generalizing existing frameworks for TSP on QAP. Moreover, in~\cite{bagga2023solving} the authors  leverage a novel double pointer network, which alternates between selecting a location for the next facility and choosing a facility for the previously selected location. Although the effectiveness of this paradigm has been demonstrated, it is limited to standard QAP that requires an equal number of facilities and locations. Consequently, the algorithm’s applicability is restricted. In summary, a low-complexity algorithm for addressing the GQAP remains an open research area.

\section{System Model and Pre-deployment}
\label{sec: systemmodel}

\subsection{System Overview}
\label{sec: overview}
In this work, we consider a single-cell scenario with a cluster of users running the \textit{customized} on-device applications based on SDMs, as shown in Fig.~\ref{sec: systemmodel}.  Meanwhile, for users within a cluster, we assume their PCS tasks share similarities, such as personalized dog object generation, personalized female face generation, or similar personalized style. 
Moreover, the devices of the users, such as desktops, laptops, and smart mobile devices,   are characterized by diverse computational power. To reduce computational latency and facilitate the application proliferation, users can access edge offloading
services. However, user-specific model deployment  at the edge leads to considerable
storage overhead. To this end, an efficient offloading manner is considered in this work.
Specifically, these customized on-device applications are assumed to share a cluster-wide model hosted on the edge server.   This model, derived through cluster-oriented fine-tuning, enables efficient parallel processing of tasks with batching techniques. A feasible pre-deployment approach is detailed in Section~\ref{sec: pre-de}. 

To maximize overall system efficiency and performance, we adopt a {periodic \& centralized}  offloading mechanism.  We assume that the offloading strategy is executed at fixed intervals of duration ${\Delta}$, each consisting of $K$ time slots, where each time slot is indexed by $k$. Moreover, we denote the set of users sending requests during ${\Delta}$ as $\mathcal{I}$, with each user indexed by $i$. A user $i \in \mathcal{I}$ can send an offloading request at any time slot $k$, but must wait $(K-k)$ time slots to receive the server's feedback. This time interval is referred to as round-trip time (RTT), denoted by $\ell^{\rm R}$.  Moreover, we denote the number of available GPUs for the current user set $\mathcal{I}$ by $G$, which may vary over time.  We define that the offloading decisions made by the edge server for user requests include two options: \textit{grant} and \textit{deny}~\cite{liu2023resource}.  For the grant decision, the \textit{number of offloaded denoising steps} is required to be identified at the same time. All the decisions for the users in $\mathcal{I}$ are  jointly determined by four main factors: 1) the computational capacity of the users' devices; 2) the preferences for personalized accuracy  and latency; 3) the time of the request; and 4) the  current available resources at the edge. 
As with any trade-off, the degree of inference offloading inversely impacts the extent of degradation in meeting users' personalized requirements. In this sense, these factors are considered to achieve a desired trade-off between  latency and the final personalized accuracy for all users.  The metrics about  latency and personalized accuracy are specified in Section~\ref{sec: metric}. 

Once the offloading strategy is determined, all denoising steps for users who are denied are performed locally on their devices. For users who are granted, they first need to transmit their \textit{individual} prompts to the edge server. The edge server, based on the offloading strategy results for each user, performs the user-specific number of denoising steps at the edge for each user’s prompt, respectively. The intermediate results are then fedback to the users to complete the remaining offloading locally. (The hybrid inference model is detailed in Section~\ref{sec: hybird}.) For the uplink and downlink transmission involved in the process, we assume that  the spectral efficiency ${\eta}$ (or data rate per unit bandwidth) for user 
$i$ is fixed, with a given modulation and coding scheme.
Given that the transmission latency is remarkably lower than the computational  latency, we  simplify the bandwidth allocation. We assume that during the transmission process, all users granted for offloading equally share the system's reserved maximum bandwidth ${W_{\max }}$.

\begin{figure*}[t]
    \centering
    \includegraphics[width=0.9\linewidth]{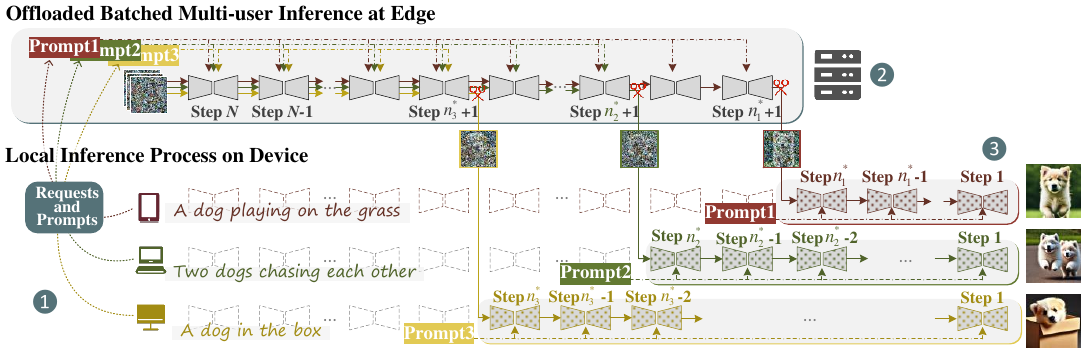}
    \caption{\color{purple}Illustration of the hybrid inference process with $n_i^*$ optimized individually for each user.}
    \label{fig:inferencemodel}
    \vspace{-0.5cm}
\end{figure*}

\subsection{Pre-deployment Approach}
\label{sec: pre-de}
Subsequently, we introduce a pre-deployment approach for the considered multi-user scenario as stated in Section~\ref{sec: overview}. For illustration purposes,  we consider the SDM-based on-device applications with \textit{personalized object generation}. This approach, however, is equally applicable to personalization tasks involving \textit{facial features}, \textit{style customization}, etc., supporting a variety of emerging diffusion-based applications such as extended reality, advertising, and remote education.

For example, as shown in Fig.~\ref{fig:systemmodel}, the term ``dog" for the available SDM refers broadly to all breeds and ages of dogs. In the absence of additional specific adjectives, the generated image of a dog is random. However, in practice, each user has a specific and largely consistent preference for the appearance of a dog. Therefore, users would prefer the term ``dog" to directly generate an image that aligns with their envisioned appearance, without the need  to input repetitive and redundant descriptive prompts each time. 
To achieve customized configurations, we assume the application uses the Dreambooth technique to fine-tune the downloaded SDM with private images. For example, the model can be fine-tuned to identify a dog as the user's specific pet, such as a golden retriever puppy.  Meanwhile, the cluster-wide model at the edge is fine-tuned using an image collection representative of a specific user cluster. These images can be either shared by users or generated based on prompts provided by them. For example, as illustrated in Fig.~\ref{fig:systemmodel}, in a cluster focused on various breeds of puppies, the data collection for fine-tuning consists of a selection of Samoyed and Golden puppy images shared by users and images generated by the prompt corgi puppy.  Then, the fine-tuned cluster-wide model can represent a dog's image as a shared representation common to different breeds of puppies. \textcolor{blue}{\textit{A case study on this, along with the visualization results, are presented in Section~\ref{sec:simulation}.}}


\begin{figure}[t]
    \centering
    \includegraphics[width=0.9\linewidth]{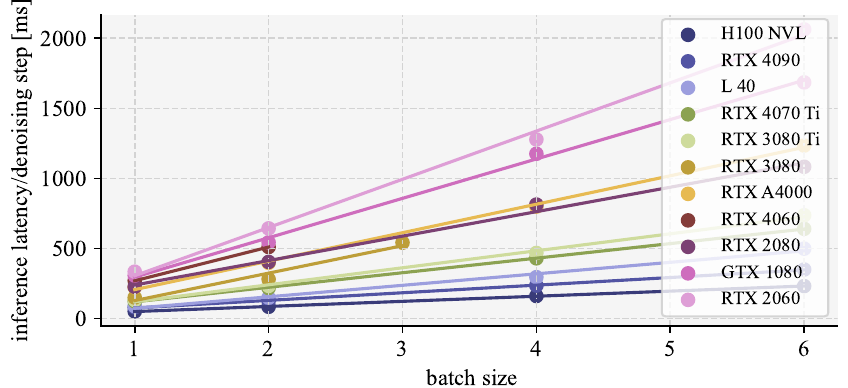}
    \caption{Illustration of the impact of batch size on latency per denoising step, based on empirical measurements.}
    \label{fig:gpu}
    \vspace{-0.5cm}
\end{figure}
\subsection{Hybrid Inference Model}
\label{sec: hybird}

Although the SDM consists of three components: \textit{Text encoder}, \textit{U-Net}, and \textit{Variational Autoencoder}, the computational complexity of the U-Net is the most significant, and it repeatedly participates in the denoising process iterations~\cite{kim2023bk}. In this sense,  we disregard the latency of the Text Encoder and Variational Autoencoder, assuming that both the  inference latency at the edge and local sides are mainly governed by their own number of denoising steps. Notably, the latency per denoising step differs between the edge and local sides, influencing their overall inference latency.

In contrast to traditional DNNs, such as convolutional neural networks, fully connected networks, and Transformer variants, which features a \textit{straightforward}, \textit{single forward pass} for inference, diffusion models employ an iterative inference process. Specifically, starting with a noise vector, the model progressively applies a series of denoising steps, repeatedly utilizing the same DNN, such as U-Net, to incrementally generate a high-quality output.
The inherently recurrent nature of the diffusion model's inference process challenges the effectiveness of existing partition-based offloading strategies for DNNs. To address the associated inference latency, this work proposes a hybrid inference model that partitions the iterative inference process in the temporal domain, as illustrated in Fig.~\ref{fig:inferencemodel}. 

Numerous simulation experiments indicate that the initial denoising steps primarily capture common low-level features, such as contours and blobs, with minimal inclusion of personalized information. Consequently, as shown in the bottom of Fig.~\ref{fig:inferencemodel}, users can offload the \textit{initial} denoising steps to a cluster-wide model, followed by a customized denoising process locally. Let the total number of denoising steps be denoted by $N$. Since the inference process reverses the sequence from the original data back to pure noise, the denoising steps are indexed in descending order. Thus, step $N$ is executed first, and step 1 is executed at last. Moreover, we assume that the last $\hat N$ denoising steps must be executed locally to meet users' personalized requirements. The value of $\hat N$ is determined experimentally in Section~\ref{sec:simulation}. Thereby, with the split point $n_i^* \in [ {\hat N,N} ]$, the denoising steps from $N$ to $n_i^* + 1$ are executed at the edge, and the remaining denoising steps from $n_i^*$ to $1$ are performed locally. For an extreme case, $n_i^* = N$ indicates that the entire inference process is performed locally.

Moreover, as shown in the top of Fig.~\ref{fig:inferencemodel}, To optimize the offloading process at the edge, the batching technique is employed to harness parallel computing capabilities and reduce memory access time. Specifically, at intervals of duration $\Delta$,  the edge server assembles the pure noisy data, denoising step indexes, and prompts of all the granted users, and batches them separately. Then, the formed three individual batches are fed as three distinct inputs into the model for parallel and iterative inference. Due to varying offloading task volumes among different users, some tasks may complete earlier, resulting in a reduced batch size.  
According to the NVIDIA report\footnote{https://www.nvidia.com/content/dam/en-zz/Solutions/Data-Center/tesla-product-literature/t4-inference-print-update-inference-tech-overview-final.pdf}, when utilizing the batching technique, the function of the computational latency on the batch size for a specific neural network is observed to be approximately linear, usually with a large constant term~\cite{liu2023resource}. In this context, we use ${{\mathscr L}}\left( b  \right) = {\rm k }b  + {\rm h }$  to characterize the computational latency of \textit{one denoising step},  where $b$ represents the batch size. Meanwhile, $\rm k$ and $\rm h$ are constant terms whose values depend on the adopted models of GPUs, which can be determined experimentally as shown in Fig.~\ref{fig:gpu}. Hereinafter, for the sake of simplicity, we use $\left( {{{\rm{k}}_i},{{\rm{h}}_i}} \right)$ and $\left( {{{\rm{k}}_e},{{\rm{h}}_e}} \right)$ to characterize the computation power of user $i$ and edge server, respectively. For the local inference, we can default the value of $b$ to always be 1. Then, we denote the latency per local denoising step for user $i$ by $\mathscr{L}_i\left(1\right)$ for short.
Furthermore, considering the commonly used data partitioning approach in parallel inference, we assume that with $G$ GPUs on the edge server, the inference latency can be approximately reduced to the latency corresponding to a batch size of $1/G$ of the original, that is, $\mathscr{L}\left( b \right)$ can be further modified as $\mathscr{L}_e\left( b, G \right) \approx {\rm k}_e\left( {b/G} \right) + {\rm h}_e$.  




\subsection{Quality of Experience Metric}
\label{sec: metric}

\subsubsection{End-to-end latency}
\label{sec:letancy}
The end-to-end (E2E) latency for a stable diffusion task is measured from the moment the user sends the offloading request to the edge server until the final generated image is obtained. As stated in Section~\ref{sec: overview}, the edge server responds to requests in two ways, each leading to a different composition of E2E latency. For users whose requests are denied, the E2E latency is equivalent to the sum of RTT and the latency of  full local inference, expressed as ${L_i} = \ell_i^{\rm{R}} + \ell_i^{\text{L}}$. For users whose requests are granted, the E2E latency includes two additional components,  which is given by ${L_i} = \ell_i^{\rm{R}} + \ell_i^{\text{E}} + \ell_i^{\text{T}} + \ell_i^{\text{L}}$, where $\ell_i^{\text{E}}$ denotes the inference latency at the edge, and $\ell_i^{\text{T}}$ represents the total transmission latency, including both uplink and downlink.  The detailed E2E latency calculations under a given strategy policy are provided in Section~\ref{sec: formulation}. 
\subsubsection{Personalized accuracy index}
In the context of efficient diffusion offloading, we introduce a tailored metric called the Personalized Accuracy Index (PAI), which integrates both personalization and accuracy. To optimize the offloading strategy, we establish the relation between PAI and the optimal split point $n^*$. Firstly, {\color{blue}for personalized accuracy measurement}, we leverage the Contrastive Language–Image Pre-training (CLIP) model\footnote{https://github.com/openai/CLIP .} to assess the the similarity between  the generated images and the user's prompt \textit{\color{blue}with additional personalized descriptions concerning the targeted object, style, and other relevant aspects}, denoted as ${\text{CLIP}}\left( {{ \text{p}_i^+,\text{m}_i}} \right)$.
{\color{blue}To further focus on personalized performance}, we use the Learned Perceptual Image Patch Similarity (LPIPS) model\footnote{https://github.com/richzhang/PerceptualSimilarity .} to measure differences in the final images generated by individual users, which is expressed by $\overline {{\rm{LPIPS}}} = \mathbb{E}_{1 \leq i < {i'} \leq |\mathcal{I}|} \left[\text{LPIPS}\left( \text{m}_i, \text{m}_{i'} \right)\right]$. These images denoted by $\text{m}_i, \text{m}_{i'}$ are based on the same intermediate result provided by the cluster-wide model at the edge server, with each user applying their customized models for subsequent denoising inference. In this sense, a greater difference between the images generated by different users indicates stronger personalization. \textcolor{blue}{Based on the fine-tuned cluster-wide model and local personalized model presented in Section~\ref{sec: pre-de}, we conducted extensive simulation experiments for hybrid inference by offloading varying numbers of denoising steps to the edge server, \textit{which is detailed in Section~\ref{sec:simulation}}}.
\textcolor{purple}{From the simulation results}, we can observe that,  as the split point decreases (the number of offloading steps increases), both the accuracy and personalization metrics exhibit a downward trend, with the decline in personalization being more pronounced. \textcolor{blue}{However, this trend is not gradual. As the number of offloaded denoising steps increases, the personalized features of each user remain well-preserved within a certain range. Yet, once a critical threshold--approximately when the offloaded steps reach or exceed half of the total denoising steps--is crossed, the degradation of these features accelerates significantly. With excessive offloading, the generated images across different users become nearly indistinguishable, ultimately losing their distinct personalized characteristics. This phenomenon further confirms the necessity of setting the value of $\hat N$ as stated in Section~\ref{sec: hybird}}.

Moreover, to capture this effect, we modulate the $\overline {{\rm{LPIPS}}}$  values using a sigmoid function, $\sigma_{ {{\rm{P}}}} \left( x \right) = \frac{1}{{1 + \exp \left( { - a\left( {x - b} \right)} \right)}}$, where $x = \overline {{\rm{LPIPS}}}$ with varying split point, $b$ represents the threshold at which personalization undergoes a significant shift, and 
$a$ indicates the slope of this change. Based on the range of $\overline {{\rm{LPIPS}}}$ values and our observations, we set 
$a=30$ and $b=0.1$.
Considering the interdependence and mutual reinforcement between accuracy and personalization, the PAI can be expressed by 
\begin{equation}
\text{PAI} = \kappa \cdot \mathbb{E}_{i \in {I_t}} \left[\text{CLIP}\left( \text{p}_i^{+}, \text{m}_i \right)\right] \cdot {\sigma _{\rm{P}}}\left( {\overline {{\rm{LPIPS}}} } \right), \label{eq:PAI}
\end{equation}
where $\kappa = 3$. 
Similarly, we also use the sigmoid fitting approach, denoted by $\mathscr{F}(\cdot)$,  \textcolor{blue}{with the fitting parameters $a_{\mathscr{F}} >0$ and $b_{\mathscr{F}}>0$}.
This function is utilized in Section~\ref{sec: formulation} to formulate the objective function that guides the optimization of the offloading strategy and \textcolor{blue}{is detailed in Section~\ref{sec:simulation}}. 
\section{Problem Formulation for Efficient Edge-assisted Offloading} 
\label{sec: formulation} 

As stated in Section~\ref{sec: hybird}, we assume that the edge server allocates $G$ GPUs to the set $\mathcal{I}$ of users who have sent requests within the past  time period $\Delta$. Therefore, we assume that the current offloading will not impact the performance of previous offloading strategies, and likewise, the performance of the current offloading strategy will not be affected by subsequent offloading. In this sense,  we mainly focus on one snapshot of personalized diffusion model offloading without loss of generality.


Before presenting the problem formulation, we first represent the outcome of the offloading strategy  by a one-hot matrix ${{\bf{x}}^{\left| {\cal I} \right| \times \left| {\cal D} \right|}}$ and an integer index vector ${{\bf{n}}^*}^{\left| {\cal I} \right| \times 1}$. Specifically,
we define two binary decision variables, $x_i^{\rm{g}}$  and $ x_i^{\rm{d}}$,  corresponding to \textit{grant} and \textit{deny}, respectively. For simplicity in subsequent
notation, we   abbreviate  as $x_i^j$, $j \in {\cal D} = \left\{ {{\rm{g}},{\rm{d}}} \right\}$. If $x_i^j = 1$, it indicates that the request handling corresponding to $j$ is selected for user $i$. In this sense, we should ensure 
\begin{equation}
   {\text{C1: }} \quad \sum\limits_{j \in {\cal D}} {x_i^j = 1}, \quad 
 \forall i \in {\cal I}, \label{eq:c1}
\end{equation} 
which signifies the
condition that each request has only one handling approach. Meanwhile, considering computational and communication efficiency as well as stability, we assume that the edge server imposes a concurrent user limit for the number of users performing offloading in a single round, denoted $B_{\max}$. In this sense, we have a common   constraint in offloading tasks, as below,
\begin{equation}
    {\text{C2: }} \quad  \sum\limits_{i \in {\cal I}} {x_i^{\rm{g}}}  \le B_{\max}. \label{eq: c2}
\end{equation}
In addition, for the spilt point determined by the offloading strategy for user~$i$, we denote it as $n_i^*$, which is consistent with the definition of $n^*$ in Section~\ref{sec: metric}. Therefore, we have 
\begin{equation}
    {\text{C3: }}\quad \hat N \leqslant n_i^* \leqslant N, \quad    \forall i \in {\mathcal{I}}. \label{eq: c3}
\end{equation}
Moreover, given the coupled relationship between  decision variables $ {\mathbf{x}}$ and ${{{\mathbf{n}}^*}}$, an additional constrain is required as below,
\begin{equation}
     {\text{C4:  }} \quad {{{x_{i}^{{\text{d}}}}}} \cdot \left( {N - n_i^*} \right){\text{  =  0,}} \quad \forall i \in {\mathcal{I}}, \label{eq:c4}
\end{equation}
which indicates that for users who do not offload, all denoising steps are completed locally, with the split point determined as $n_i^* = N$.

After defining the decision variables and their domains, we proceed to detail the objective function   based on the QoE metrics defined in Section~\ref{sec: metric}. First, we assign each user an individual weight parameter, denoted as 
${\alpha _i}$, which characterizes their emphasis on latency and PAI. The value of ${\alpha _i}$ can vary for each user. Therefore, the QoE of user $i$ can be comprehensively evaluated as
${{\alpha _i}{\cal F}\left( {n_i^*} \right) - {L_i}}$. 
According to the definitions of the two forms of E2E latency in Section~\ref{sec:letancy}, by integrating the decision variables  ${\bf{x}}$ into $L_i$, $L_i$
can be uniformly expressed as 
  \begin{equation}
      {L_i} = {\ell ^{\rm{R}}} + {\ell ^{\rm{T}}} \cdot x_i^{\rm{g}} + {\ell ^{\rm{E}}} \cdot x_i^{\rm{g}} + {\ell ^{\rm{L}}}.
  \end{equation}
Therein, ${\ell ^{\rm{R}}} = K -k$, which is consistent for all users regardless of the different request handling they are assigned with. Meanwhile, ${\ell ^{\rm{L}}}$ is also a common term for the E2E latency corresponding to different request handling results, but its representation varies. For $x_i^{\rm d} = 1$, the expression of ${\ell ^{\rm{L}}}$ is fixed as ${\ell ^{\rm{L}}} = N{{\mathscr L}_i}\left( 1 \right)$, while for $x_i^{\rm g} = 1$, its expression may vary depending on the user. With the constraint    \eqref{eq:c4} in mind, ${\ell ^{\rm{L}}}$ can be uniformly expressed as a function of $n_i^*$, given by
${\ell ^{\rm{L}}} = n_i^*{{\mathscr L}_i}\left( 1 \right)$. In contrast to ${\ell ^{\rm{R}}}$ and ${\ell ^{\rm{L}}}$, terms ${\ell ^{\rm{T}}}$ and ${\ell ^{\rm{E}}}$ only exist  in the E2E latency for users whose requests are granted. Specifically, for $\ell^{\rm T}$, it consists of  the latency for both the uplink transmission of the prompt and the downlink transmission for the intermediate result.
We denote the data sizes of the prompt and intermediate noisy data by $s^{\rm p}_i$, and $s^{\rm m}_{i}$, respectively. Although JPEG encoding may cause different split points to produce varying levels of image noisiness, which can slightly affect the data size of intermediate results, these differences have a negligible impact on the E2E latency. Therefore, we assume the data size of the intermediate results to be constant.  Then, $\ell_i^{\rm T}$ for the granted users with $x_i^{\rm{g}} = 1$ can be calculated by 
\begin{equation}
\ell _i^{\text{T}} =  {\tfrac{{s_i^{\text{p}}}}{{{\eta W_i^{\text{g}}}}} + \tfrac{{s_{i}^{{\text{m}}}}}{{{\eta W_i^\text{g}}}}}, \label{eq:lit}
\end{equation}
 where $W_i^{\rm g}$ represents the bandwidth allocated to user $i$. According to the straightforward bandwidth allocation described in Section~\ref{sec: overview}, where granted users share the bandwidth equally, we have 
 \begin{equation}
   {W_i^{\rm g}} = \tfrac{{{W_{\max }}}}{{\sum\nolimits_{i \in \mathcal{I}} {{x_i^{\text{g}}}} }}. \label{eq:w}
 \end{equation}
The calculation of $\ell_i^{\rm E}$ shares certain similarity with that of $\ell_i^{\rm T}$.
However, calculating the computational latency for edge offloading presents more complications. Due to the varying split points $n_i^*$ for each user, the batch size may change during the offloaded inference process, as shown in Fig.~\ref{fig:inferencemodel}, leading to potentially different computational latency for each denoising step. Nonetheless, considering that the number of denoising steps predominantly influences the overall latency, we simplify the modeling of the batch size's impact on latency. The batch size during the offloading at the edge is assumed to be fixed as  
\begin{equation}
    b_{\rm{g}}  =   \sum\limits_{i \in \mathcal{I}} {{x_{i}^{{\text{g}}}}}. \label{eq:lbo}
\end{equation}
 This simplification neglects the impact of certain tasks completing early on the reduction of the denoising step latency for the remaining tasks.  With the batch size to the maximum value, it provides a conservative estimate of the worst-case latency. In this context,  $\ell_i^{\rm E}$ can be expressed by 
\begin{equation}
\ell_i^{\rm E} = \left( {N - n_i^*} \right) {\mathscr L_e}\left( b_{\text g}, G \right).    \label{eq:leo} 
\end{equation}
In light of above, since we aim to optimize the overall performance of all the users, the optimization problem of the offloading strategy can be formulated as  
\begin{equation}
    \mathop {\max }\limits_{{\bf{x}},{{\bf{n}}^*}} \sum\limits_{i \in {{\cal I}}} {{\alpha _i}{\mathscr{F} }\left( {n_i^*} \right) - {L_i}} \label{P1} \tag{P1}
\end{equation}
subject to
\begin{equation}
  \eqref{eq:c1}, \eqref{eq: c2}, \eqref{eq: c3}, \eqref{eq:c4}.  \nonumber
\end{equation}

\begin{proposition}
The multi-dimensional coupled integer programming problem shown in \eqref{P1} is NP-hard.
\end{proposition}

\begin{proof}
As shown in \eqref{P1}, there are two sets of coupled  optimization variables, ${\mathbf{x}}$ and ${{\mathbf{n}}^*}$ with high dimensions of ${\left| {{\mathcal{I}}} \right| \times \left| \mathcal{D} \right|}$ and ${\left| {{\mathcal{I}}} \right| \times 1}$, respectively. To demonstrate that \eqref{P1} is an NP-hard problem, we first simplify it into the form of a standard NP-hard problem, ensuring that the fundamental complexity and core structure remain unchanged. Then, based on the transitivity of NP-hardness, we infer that the original problem is also NP-hard. Specifically,  we fix the value of the decision variable $n_i^*$ by $n_i^* = \hat N$ for the granted users, assuming that these users do not have high demands for image quality and are primarily focused on reducing latency. Once users opt for offloading, they offload the maximum possible number of denoising steps, i.e., $( {N - \hat N} )$. In this sense, the PAI metric $\sum\nolimits_{i \in \mathcal{I}} {\alpha_i \mathscr{F}\left( {n_i^*} \right)} $ can be reduced to a function dependent solely on ${\mathbf{x}}$, i.e., 
\begin{equation}
\sum\limits_{i \in \mathcal{I}} {\alpha_i \mathscr{F}\left( n_i^* \right)}= \sum\limits_{i \in \mathcal{I}} {\sum\limits_{j \in \mathcal{D}} {{C_{i
j}}{x_{i}^{j}}} }, \label{eq: p_performance_C}
\end{equation}
where  ${C_{i}^{\rm {d}}} = \alpha_i\mathscr{F}( N)$ and ${C_{i}^{\text{g}}} = \alpha_i\mathscr{F}(\hat N)$. Meanwhile, C3 and C4 are removed. Moreover, for the penalty term related to latency, substituting  \eqref{eq:w} into  \eqref{eq:lit}, we have 
\begin{equation}
   \ell _i^T = \sum\limits_{i' \in \mathcal{I}} {{A_{i'\text{g}i{\text{g}}}^\text{T}}\cdot{x_{i'}^{\text{g}} } \cdot {x_{i}^{{\text{g}}}}} , 
\end{equation}
where ${A_{i'\text{g}i{\text{g}}}^\text{T}}= \frac{{s_i^{\text{p}} + s_{i}^{{\text{m}}}}}{{{\eta }{W_{\max }}}}$. Considering the overall transmission latency of the multi-user system, we have 
\begin{equation}
    {L^{\text T}} = \sum\limits_{j \in \mathcal{D}} {\sum\limits_{j' \in \mathcal{D}} {\sum\limits_{i \in \mathcal{I}} {\sum\limits_{i' \in \mathcal{I}} {{A_{i'j'ij}^\text{T}}\cdot{x_{i'}^{j'}} \cdot {x_{i}^{j}}} } } } , \label{eq: LT}
\end{equation}
where ${A_{i'j'ij}^\text{T}} = 0$, if $j = \rm{d}$ or $j' = \rm{d}$. Similarly, substituting \eqref{eq:lbo}  into \eqref{eq:leo},   and explicitly expanding the form of the function  $\mathscr L_e (b,G) = {\rm{k}}_e(b /G) + {\rm{h}}_e$, we have 
\begin{equation}
\ell _i^{\text{E}} = D_{i{\rm{g}}}^{\text{E}} \cdot {x_{i}^{{\rm{g}}}} + \sum\limits_{i' \in \mathcal{I}} {A_{i'{\rm{g}}i{\rm{g}}}^{\text{E}} \cdot {x_{i'}^{{\rm{g}}}} \cdot {x_{i}^{{\rm{g}}}}},
\end{equation}
where $D_{i{\rm{g}}}^{\text{E}} = ( {N - \hat N} ){{{\rm h}_e}} $ and $A_{i'{\rm{g}}i{\rm{g}}}^{\text{E}} =  (N - \hat N){{{{\rm{k}}_e}} \mathord{\left/
 {\vphantom {{{{\rm{k}}_e}} G}} \right.
 \kern-\nulldelimiterspace} G}$. When extending to multi-user system, we have 
\begin{equation}
    {L^{\text{E}}} = \sum\limits_{i \in \mathcal{I}} {\sum\limits_{j \in \mathcal{D}} {D_{ij}^{\text{E}}} }  \cdot {x_{i}^{j}} + \sum\limits_{j \in \mathcal{D}} {\sum\limits_{j' \in \mathcal{D}} {\sum\limits_{i \in \mathcal{I}} {\sum\limits_{i' \in \mathcal{I}} {A_{i'j'ij}^{\text{E}} \cdot {x_{i'}^{j'}} \cdot {x_{i}^{j}}} } } } , \label{eq: LE}
\end{equation}
where $D_{i\text{d}}^{\text E} = 0$ and $A_{i'j'ij}^{\text{E} }=0$, except for  $A_{i'\text{g}i{\text{g}}}^{\text{E}}$. In addition, for the latency for local inference, similarly, we have 
\begin{equation}
{L^{\text L}} = \sum\limits_{j \in {\cal D}} {\sum\limits_{i \in {\cal I}} {D_{ij}^{\text L} \cdot x_i^j} }, \label{eq: LL}
\end{equation}
where $D_{i{\rm{d}}}^{\text{L}} = N({{\rm k}_i} + {{\rm h}_i})$ and $D_{i{\rm{g}}}^{\text{L}}  = 
{\hat N}({{\rm k}_i} + {{\rm h}_i})$. 
At last, for $L^{\text R}$, we have 
\begin{equation}
    {L^{\rm{R}}} = \sum\limits_{j \in {\cal D}} {\sum\limits_{i \in {\cal I}} {D_{ij}^{\rm{R}} \cdot x_i^j} }, \label{eq: LR} 
\end{equation}
where ${D_{ij}^{\text{R}}} = K - k_i$, and $k_i$ denotes the moment when user~$i$ sends the request to the edge, as stated in Section~\ref{sec: overview}.

Overall, with \eqref{eq: p_performance_C}, \eqref{eq: LT}, \eqref{eq: LE}, 
\eqref{eq: LL}, and \eqref{eq: LR}, the objective function of~\eqref{P1} can be rewritten as 
\begin{equation}
   \mathop {\max }\limits_{\mathbf{x}} \sum\limits_{i \in \mathcal{I}} {\sum\limits_{j \in \mathcal{J}} {{D_{ij}} \cdot {x_{i}^{j}}}  - } \sum\limits_{j \in \mathcal{D}} {\sum\limits_{j' \in \mathcal{D}} {\sum\limits_{i \in \mathcal{I}} {\sum\limits_{i' \in \mathcal{I}} {A_{i'j'ij}^{} \cdot {x_{i'}^{j'}}  {x_{i}^{j}}} } } },  \label{eq: re-obj}
\end{equation}
where ${D_{ij}} = {C_{ij}} - \left( {D_{ij}^{\rm{L}} + D_{ij}^{\rm{E}} + D_{ij}^{\rm{R}}} \right)$ and $A_{i'j'ij} = A_{i'j'ij}^{\text{T}} + A_{i'j'ij}^{\text{E}}$. According to \cite[Theorem 1]{hammer1970some}, the linear term can be  removed from \eqref{eq: re-obj} as the diagonal of matrix of quadratic terms. Then, \eqref{eq: re-obj} can be reformulated as 
\begin{equation}
    \mathop {\max }\limits_{\mathbf{x}} \sum\limits_{j \in \mathcal{D}} {\sum\limits_{j' \in \mathcal{D}} {\sum\limits_{i \in \mathcal{I}} {\sum\limits_{i' \in \mathcal{I}} {\tilde A_{i'j'ij}^{} \cdot {x_{i'}^{j'}} \cdot {x_{i}^{j}}} } } },  \label{eq:simplest}
\end{equation}
where $\tilde A_{ijij} = D_{ij} - A_{ijij}$  while all other elements of $\tilde A_{i'j'ij} = -  A_{i'j'ij}, i' \ne i, j' \ne j$ remain unchanged. If $\left| {\cal D} \right| = \left| {\cal I} \right|$, the optimization problem in \eqref{eq:simplest} is is a QAP. However, in the  vast majority of cases, $\left| {\cal D} \right| \ll \left| {\cal I} \right|$. Additionally, considering the remaining constraint 
 C2 in \eqref{eq: c2} from the original problem \eqref{P1}, the simplified version can be more accurately described as a GQAP, widely proven to be NP-hard~\cite{fisher1986multiplier,sahni1976p}
Since the reduced problem~\eqref{eq:simplest} remains NP-hard, the original~\eqref{P1} includes additional constraints and variables, making it more complex, it is reasonable to infer that~\eqref{P1} is at least NP-hard.
\end{proof}

\section{DRL-Convex Hybrid Solution for Efficient Edge-assisted Offloading}
\label{sec: DRL-COnvex}
To address the GQAP {\color{blue}about $\mathbf{x}$}, further complicated by the additional optimization of a coupled decision variable of integer index vector, {\color{blue}$\mathbf{\mathbf{n^*}}$}, in this section, we propose a novel DRL-convex hybrid approach.  {\color{blue}Specifically, to make the extended version of GQAP tractable}, we first transform it into two nested sub-problems with $\mathbf{x}$ and $\mathbf{n}^*$, respectively, which is detailed in Section~\ref{sec:nest}.
\subsection{{\color{blue}Insights for} Nested Sub-Problem Formulation}
\label{sec:nest}
Recall \eqref{P1}: the optimal values of $\mathbf{x}$ and $\mathbf{n}^*$ are interdependent. {\color{blue}The total number of granted users, determined by $\mathbf{x}$, influences the batch size during the inference at the edge server, which in turn affects the latency of each offloaded denoising step. This leaves a customized trade-off between PAI and E2E latency. It is shaped by each user's local computational capacities $\left( {{{\rm{k}}_i},{{\rm{h}}_i}} \right)$ and emphasis factors $\alpha_i$, ultimately impacting the optimal value of $n_i^*$.
Meanwhile, for each user $i$, the choice of $n_i^*$  directly affects their local inference latency $\ell_i^{\rm L}$ and offloaded inference latency $\ell_i^{\rm E}$. For a specific emphasis factor  $\alpha_i$, each user's end-to-end latency and PAI jointly determine the overall performance of the whole scenario. Given the shared resources at the edge, it in turn influences the optimal decision variable~$\mathbf{x}$.

}

Nevertheless, through an in-depth analysis of {\color{blue}the coupling nature}, the  variable $\mathbf{x}$ {\color{blue}emerges as the dominant decision variable between the two. Given a specific value of $\mathbf{x}$, the inference latency of each denoising step at the edge server becomes determined. Then, the optimization over 
$n_i^*$ can be decoupled across users and solved with a well-defined closed-form expression with the consideration of their local computational capacities $\left( {{{\rm{k}}_i},{{\rm{h}}_i}} \right)$ and emphasis factors $\alpha_i$. Moreover, each individual $n_i^*$ collectively contributes to the globally optimal $\mathbf{n}^*$. In other words, for a given $\mathbf{x}$, it is relatively straightforward to derive the optimal $\mathbf{n}^*$ that maximizes system performance under the given $\mathbf{x}$.} However, $\mathbf{n}^*$ cannot be pre-set to readily determine $\mathbf{x}$, {\color{blue} as it does not reduce the coupling effect among users. This leads to an asymmetric dependency between $\mathbf{n}^*$ and $\mathbf{x}$}.

With this in mind, we transform the original problem into two nested sub-problems. 
{\color{blue}We consider the optimization of request handling, $\mathbf{x}$, as the outer-layer optimization, while the optimization of split point, $\mathbf{n}^*$, is treated as inner-layer optimization. Given the explicit formulation  and the unified expression of the problem about $n_i^*$ across all users under a specific $\mathbf{x}$, we employ convex optimization with continuous relaxation to solve it, as detailed in Section~\ref{sec:inner op}.
However, for a given $\mathbf{x}$, explicitly quantifying how the optimization of $\mathbf{n}^*$ contributes to the overall performance achievable by the current $\mathbf{x}$ remains challenging. Furthermore, the high dimensionality of the decision variable 
$\mathbf{x}$ adds to the complexity of the outer sub-problem. 
In this sense, we adopt a DRL-based approach to effectively capture a satisfactory solution about $\mathbf{x}$. \textit{Meanwhile, the implicit influence of $\mathbf{n}^*$
  on on the optimality of $\mathbf{x}$ is incorporated into the environment modeling within the DRL paradigm.} This is elaborated in Section~\ref{sec:outer op}. Notably, the powerful learning capability of DRL also mitigates the impact of decomposing \eqref{P1} into two sub-problems on overall optimality.}

\subsection{Convex Optimization for Split Point Selection}
\label{sec:inner op}
In this subsection, {\color{purple}we first present the formulation of inner-layer optimization about $\mathbf{n}^*$}. Then, we analyze the convexity of the optimization problem, followed by specifying the method for determining the optimal solution. Finally, based on the problem analysis, we provide insights on selecting emphasis values, {\color{purple}$\alpha_i$}, for {\color{purple}individual} users with varying computational capacities.

For the users with  $x_{i}^{\rm{d}} = 1$, we set $n_i^* = N$ directly. Moreover, we have that $L_i = \ell_i^{\rm R} + \ell_i^{\rm L}$, which is solely determined by the moment when the user sends request and the computing power of the local device. In this sense,
given a feasible $\mathbf{x}$, the inner optimization problem on the split point  only needs to {\color{purple}be carried out} for {\color{purple}the granted users for offloading, i.e., the users with $x_{i}^{\rm{g}} = 1$}.  Moreover, for the granted users, $\ell_i^{\rm T}$ is {\color{purple}solely determined} $\mathbf{x}$. {\color{purple}Therefore, once $\mathbf{x}$ is given,  $\ell_i^{\rm T}$ can be treated as a constant}.  Meanwhile, the value of $\ell_i^{\rm E}$ is joint determined by $\mathbf{x}$ and ${n}_i^*$. {\color{purple}For a given $\mathbf{x}$,  it depends only on $n_i^*$}. Therefore, to ensure clarity in the subsequent analysis, we simply $\ell_i^{\rm E}\left(n_i^*,\mathbf{x}\right)$ and $\ell_i^{\rm T}\left(\mathbf{x}\right)$   to $\ell_i^{\rm E}\left(n_i^*\right)$ and $\ell_i^{\rm T}$, respectively. In this context, the inner sub-problem can be specifically formulated as follows:
\begin{equation}
    \mathop {\max }\limits_{n_i^* \in [ {\hat N,N} ]} {\alpha _i}\mathscr{F}\left( {n_i^* 
    } \right) - \left( {{\ell^{\rm{R}}}  + {\ell^{\rm{E}}}\left( {n_i^*} \right) + {\ell^{\rm{T}}} + {\ell^{\rm{L}}}\left( {n_i^*} \right)} \right),\tag{P2} \nonumber\label{P2}
\end{equation}
{\color{purple}where $i \in \left\{ {i\left| {x_i^{\rm{g}} = 1 \wedge x_i^{\rm{d}} = 0} \right.} \right\}$}.

\begin{proposition}
The split point ${n_i^*}$ optimization is concave problem with a unique optimal solution within the domain $[\hat N,N]$.
\end{proposition}
\begin{proof}
For simplicity, we combine the constant terms that are independent of $n^*$.  Recall that $\ell_i^{\rm E} = (N-n_i^*)\mathscr{L}_e(b_{\rm g}, G)$ and $\ell_i^{\rm L} = n_i^*\mathscr{L}_i(1)$, the objective function is rewritten as
\begin{equation}
\mathscr{O}(n_i^*) = {\alpha _i}\mathscr{F}\left( {n_i^*} \right) - \left( {{\mathscr{L}_i}\left( 1 \right) - {\mathscr{L}_e}\left( {{b_{\rm g}}}, G\right)} \right)n_i^* + \Gamma,
\end{equation}
where $\Gamma  = {\ell^{\rm{R}}} + {\ell^{\rm{T}}} - N{\mathscr{L}_e}\left( {{b_{\rm g}}}, G \right)$. Recall the expression of $\mathscr{F}\left(n_i^*\right)=\frac{1}{{1 + \exp \left( { - a_{\mathscr{F}}\left( {x - b_{\mathscr{F}}} \right)} \right)}}$  as well as the second derivative of the sigmoid function, we have 
\begin{equation}
\frac{{d\mathscr{O}(n_i^*)}}{{dn_i^*}} = {\alpha _i}a_{\mathscr{F}}{\mathscr{F}}\left( {n_i^*} \right)\left( {1 - {\mathscr{F}}\left( {n_i^*} \right)} \right) - \left( {{{\mathscr{L}}_i}\left( 1 \right) - {{\mathscr{L}}_e}\left( {{b_{\rm g}}}, G \right)} \right),\nonumber
\end{equation}
\begin{equation}
    \frac{{{d^2}\mathscr{O}(n_i^*)}}{{d{{(n_i^*)}^2}}} = {\alpha _i}{a_{\mathscr{F}}^2}{\mathscr{F}}\left( {n_i^*} \right)\left( {1 - {\mathscr{F}}\left( {n_i^*} \right)} \right)\left( {1 - 2{\mathscr{F}}\left( {n_i^*} \right)} \right).\nonumber
\end{equation}
Considering that both ${\alpha _i}$ and $a_\mathscr{F}$ are positive and $\mathscr{F}(n_i^*) \in (0,1)$, the sign of 
 $\frac{{{d^2}\mathscr{O}(n_i^*)}}{{d{{(n_i^*)}^2}}}$ depends on $\left( {1 - 2\mathscr{F}\left( {n_i^*} \right)} \right)$. 
\textcolor{purple}{Since $\mathscr{F}(n^*)$ is a monotonically increasing function, we can select an appropriate value of $\hat N$ to ensure that 
$\mathscr{F}(n_i^*) > 0.5$ for all 
$n_i^* \in [ \hat N, N ]$. }
 Therefore, $\frac{{{d^2}\mathscr{O}(n_i^*)}}{{d{{(n_i^*)}^2}}}$ is constantly less than 0 within the domain. In this sense, $\mathscr{O}(n_i^*)$ is strictly concave. Within $[\hat N, N]$, \eqref{P2} has a unique optimal solution.
 
Firstly, in case of  ${{{\mathscr{L}}_i}\left( 1 \right) < {{\mathscr{L}}_e}\left( {{b_{\rm g}}}, G \right)}$, which means that the local inference latency per denoising step is lower than the inference latency at the edge server,
$\frac{{d\mathscr{O}(n_i^*)}}{{dn_i^*}}$  is strictly positive. In this context, when $n_i^* = N$, \eqref{P2} reaches its maximum value, which implies that user $i$ completes inference locally. In contrast, if ${{{\mathscr{L}}_i}\left( 1 \right) > {{\mathscr{L}}_e}\left( {{b_{\rm g}}}, G \right)}$, the solution has three possible cases. Given that $\frac{{{d^2}\mathscr{O}(n_i^*)}}{{d{{(n_i^*)}^2}}} < 0$,  $\mathscr{O}'(n_i^*) = \frac{{d\mathscr{O}(n_i^*)}}{{dn_i^*}}$  is monotonically decreasing on $[\hat N, N]$.  Denote the term of ${\alpha _i}{a_{\mathscr{F}}}{\mathscr{F}}\left( {n_i^*} \right)\left( {1 - \mathscr{F}\left( {n_i^*} \right)} \right)$ by ${\mathscr{\bar F}}'(n_i^*)$. If ${\mathscr{\bar F}}'(\hat N) > {\mathscr{\bar F}}'(N) > \left( {{{\mathscr{L}}_i}\left( 1 \right) - {{\mathscr{L}}_e}\left( {{b_{\rm g}}}, G \right)} \right)$, $\mathscr{O}(n_i^*)$ is monotonically increasing within the domain, and $n_i^* = N$. If ${\mathscr{\bar F}}'(N) < {\mathscr{\bar F}}'(\hat N) < \left( {{{\mathscr{L}}_i}\left( 1 \right) - {{\mathscr{L}}_e}\left( {{b_{\rm g}}}, G\right)}\right)$, $\mathscr{O}(n_i^*)$ is monotonically decreasing within the domain, and $n_i^* = \hat N$. In contrast, in case of  ${\mathscr{\bar F}}'( N)  < \left( {{{\mathscr{L}}_i}\left( 1 \right) - {{\mathscr{L}}_e}\left( {{b_{\rm g}}},G \right)} \right) < {\mathscr{\bar F}}'(\hat N)$, 
  there exists a value of $n_i^*$ that makes $\frac{{d\mathscr{O}(n_i^*)}}{{dn_i^*}} = 0$, which is the optimal solution. Due to the difficulty in obtaining a closed-form expression for $n_i^*$, its solution can be found using libraries for solving equations, such as  SciPy.
\end{proof}

\begin{remark}
In the optimal solution analysis above, for users whose local inference speed is lower than the inference speed of edge offloading, the value of $\alpha_i$ can be determined based on the following guidelines.
 For a given ${\alpha _i} > \frac{{\left( {{\mathscr{L}_i}\left( 1 \right) - {\mathscr{L}_e}\left( {{b_{\rm g}}},G \right)} \right)}}{{{a_\mathscr{F}}{\mathscr{F}(N)(1-\mathscr{F}(N))}}}$,  the focus of user $i$ on PAI is overwhelmingly dominant. For a given ${\alpha _i} < \frac{{\left( {{\mathscr{L}_i}\left( 1 \right) - {\mathscr{L}_e}\left( {{b_{\rm g}}},G \right)} \right)}}{{{a_\mathscr{F}}{\mathscr{F}(\hat N)(1-\mathscr{F}(\hat N))}}}$, the optimization problem effectively reduces to optimizing solely for latency. In contrast,  for
 $\frac{{\left( {{\mathscr{L}_i}\left( 1 \right) - {\mathscr{L}_e}\left( {{b_{\rm g}}},G \right)} \right)}}{{{a_\mathscr{F}}{\mathscr{F}(\hat N)(1-\mathscr{F}(\hat N))}}} < {\alpha _i} < \frac{{\left( {{\mathscr{L}_i}\left( 1 \right) - {\mathscr{L}_e}\left( {{b_{\rm g}}},G \right)} \right)}}{{{a_\mathscr{F}}{\mathscr{F}(N)(1-\mathscr{F}(N))}}}$, it can achieve a personalized  trade-off between PAI and latency, with larger values of ${\alpha _i}$ placing greater emphasis on PAI. Although the exact value of $b_{\rm g}$
  is unknown, the service provider can estimate 
${\hat b}_{\rm g}$
  based on the number of active users in the area, thereby offering each user a range for selecting $\alpha_i$.                                               Additionally, it should be noted that the latency optimization considered in the individual split point  only accounts for the computational latency. To incorporate other types of latency integrated in $\Gamma$, the exponent of the total latency term in $\mathscr{O}(n_i^*)$ can be increased accordingly.
\end{remark}
\begin{figure}[t]
    \centering
    \includegraphics[width=0.9\linewidth]{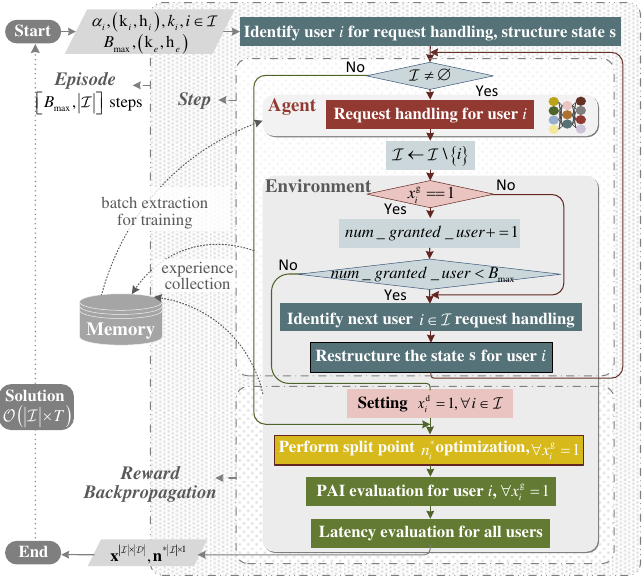}
    \caption{Flowchart for transitioning from extended GQAP to decision sequence.}
    \label{fig:flowchart}
    \vspace{-0.5cm}
\end{figure}
\subsection{DRL-based Request Handling Optimization}
\label{sec:outer op}
In this subsection, we aim to present a novel low-complexity DRL-based algorithm for the outer-layer problem of an extended GQAP.  Specifically, to solve the problem using DRL, we first formulate it as a sequential decision problem, which is detailed in Section~\ref{sec: MDP}. Subsequently, the optimal decision variable $\mathbf{x}$ is determined by progressively making step-by-step decisions that optimize the long-term reward, corresponding to the objective of \eqref{P1}. The decision-making policy for each step is introduced in Section~\ref{sec: DQN}.
\begin{figure}[t]
    \centering
    \includegraphics[width=0.9\linewidth]{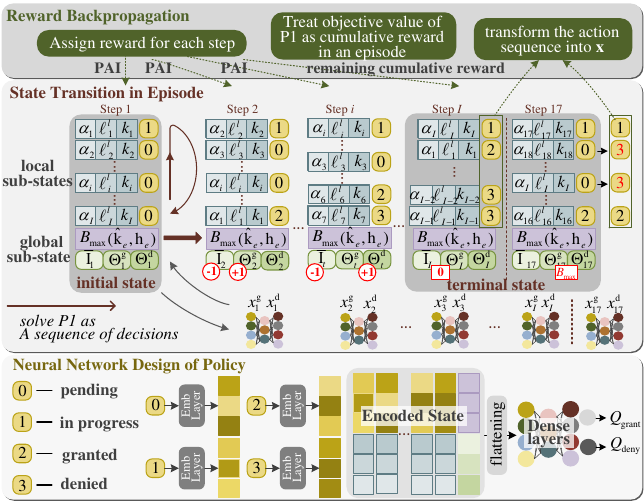}
    \caption{Illustration of DQN-based request handling.}
    \label{fig: DQN}
    \vspace{-0.5cm}
\end{figure}
\subsubsection{MDP construction}
\label{sec: MDP}
{\color{blue}As stated in Section~\ref{sec:2-3}, unlike ordering problems such as TSP, incrementally constructing a sequence of nodes is not inherent in variants of GQAP. Consequently, it cannot be straightforwardly transformed into a sequential MDP model, as is common in most existing COPs~\cite{khalil2017learning,ma2019combinatorial,cappart2021combining}. Moreover, manually applying masking to exclude previously selected actions from the action space is not a feasible way to characterize MDP transitions in this context.
Therefore, we propose a novel MDP transformation paradigm with the following key insights:
\begin{itemize}
    \item \textit{Sequential Decision Modeling}: The request handling process for each user is treated as a stepwise decision-making task, where the decoupled optimization on $n_i^*$ can integrated into the environment that the agent interacts with in each step.
    \item \textit{User State Indicator}: To capture dependencies among users, we introduce an indicator in states to track past actions for each user, the current user being processed, and pending users.
    \item \textit{Cyclic Shifting Mechanism}: This mechanism enhances the agent’s ability to distinguish between different users’ states during transitions.
    \item \textit{Reward Backpropagation}: Since rewards depend on both past and future decisions, we propagate rewards backward to ensure accurate assignment at each step.
    \item \textit{Local-Global State Design}: The MDP state consists of  local sub-states, characterizing individual users, and  global sub-states, which tracks resource constraints such as the maximum and current number of granted users.
    \item \textit{Dynamic-Static State Design}: The MDP state comprises  dynamic sub-states, capturing state transitions within an episode, and  static sub-states, randomly initialized at the episode's start to help the policy adapt to varying system configurations.
    \item \textit{Constraint-Aware Termination}: Instead of using penalty terms, episodes terminate automatically according to the global sub-state when resource constraints reach their limits, simplifying the learning process.
\end{itemize}}
The specific MDP is shown in Fig.~\ref{fig:flowchart} and described as below.

\textbf{State Space}: To facilitate the agent to provide the optimal solution for a specific scenario, it is essential to incorporate all factors related to the objective value into the state. This includes both \textit{static} information about all users and the edge server, such as computing power, and \textit{dynamic} information shaped by decisions in prior steps. Specifically, as shown in Fig.~\ref{fig: DQN}, we define the state as comprising both local and global sub-states, each further divided into static and dynamic elements. Denote the index of steps by $t$. The local sub-state for user $i$ is denoted by ${\bf{s}}_{t,i}^{\rm l} = \left[ {{\alpha _i},{\ell_i},{k_i};{o_{i,t}}} \right]$. Therein, ${\alpha _i}$ is the emphasis of user $i$ on the PAI and E2E latency,  ${k_i}$ is the moments when user $i$ sends the request to user $i$, and ${\ell_i} = {\rm k}_i + {\rm h}_i$ represents the latency per local denoising step, which is dependent on the computing power of user $i$'s device. These three elements are static during the environment evolution. In contrast, $o_{i,t}$ is the dynamic element, which acts as the indicator of four status. For users whose requests are not handled, $o_{i,t} = 0$; for the user which is identified to be in progress of request handling in the current step, $o_{i,t} = 1$; for the user whose request has been granted, $o_{i,t} =2$; and for the users whose request has been denied, $o_{i,t} =3$. 
\textit{For clarity, we denote the index of the user with  $o_{i,t} = 1$ as $i =t$.}
Meanwhile, the global state is denoted by ${\bf{s}}_{t}^{\rm{g}} = \left[ {{B_{\max }},{{\rm{\hat k}}_e},{{
 \rm{h}}_e};\bar {\rm I}_t,{\Theta _t^{\rm{g}}},{\Theta _t ^{\rm{d}}}} \right]$. Similarly, the first three elements are static, where $B_{\max}$ represents the  maximum number of users allowed to offload to the edge server in a round, and ${{\rm{\hat k}}_e}$,  $\left({{\rm{\hat k}}_e} = {{\rm{k}}_e}/G\right)$, together with ${{
 \rm{h}}_e}$ jointly characterizes the computing power of the edge server. Meanwhile, ${\bar {\rm I}}_t$, ${\Theta _t^{\rm{g}}}$, and ${\Theta _t^{\rm{d}}}$ are the dynamic elements in the global sub-state, which represent the total number of the users in pending, the total number of the granted users, and the total number of the denied users, respectively. In this context, as shown in Fig.~\ref{fig: DQN}, the complete representation of the state can  be given by ${{\mathbf{s}}_t} = [{\mathbf{s}}_{t,t}^{\text{l}},{\mathbf{s}}_{t,t + 1}^{\text{l}}, \ldots ,{\mathbf{s}}_{t,I}^{\text{l}},{\mathbf{s}}_{t,1}^{\text{l}}, \ldots ,{\mathbf{s}}_{t,i}^{\text{l}}, \ldots ,{\mathbf{s}}_{t,t - 1}^{\text{l}},{\mathbf{s}}_t^{\text{g}}]$.

\textbf{Action Space}: In each step $t$, the action of agent is to grant or deny the offloading request of user $i$ from the perspective of optimizing overall system performance. In this context, we have ${a_t} \in \left\{ {0,1} \right\}$. If $a_t = 1$, $x_t^{\rm g} = 1$ and $x_t^{\rm d} = 0$, otherwise, $x_t^{\rm d} = 1$ and $x_t^{\rm g} = 0$.

\textbf{Transition Rule}:  The state transition is determined solely by the action in a deterministic manner. Assume that in step~$t$, the request of user $t$  is in progress. Then, in the next state $\mathbf{s}_{t+1}$, $\bar {\rm I}_{t+1} = \bar {\rm I}_{t} - 1$. Moreover, if $a_t = 1$, $o_{t,{t+1}} = 2$ and ${\Theta ^{\rm{g}}_{t+1}}$ = ${\Theta ^{\rm{g}}_t}$ +1, otherwise $o_{t,{t+1}} = 3$ and ${\Theta ^{\rm{d}}_{t+1}}$ = ${\Theta ^{\rm{d}}_t}$ +1. Meanwhile, the indicator of the user for request handling in the next step is set to 1, i.e., $o_{t+1,t+1} = 1$. Moreover, to enable the agent to more easily capture the distinct characteristics of the currently selected user, we apply cyclic shifting as shown in Fig.~\ref{fig: DQN} to consistently position this sub-state at the beginning of the local sub-state. Additionally, there are two conditions for the end of an episode during the transition process: 1) if, in a state, ${\bar {\rm I}}_t =0$, it signifies that all users' requests have been processed, then the episode ends 2) if, in a state, the value of ${\Theta ^{\rm{g}}_t}$ = $B_{\max}$, it indicates that the maximum limit of granted users has been reached, then the remaining users' requests are directly considered as denied and the episode ends. That is, $o_t,i = 3$ for $i>t$.

\textbf{Reward}: Due to the batching technique, the total number of granted users $b_{\rm g}$  determines the inference latency of each offloaded denoising step across all  granted user. {\color{blue}This, in turn, affects the individual trade-off between E2E latency and PAI, as discussed in Section~\ref{sec:inner op}. Moreover, the reward for processing a user's request at each step depends not only on past decisions, which are embedded in the state representation, but also on future decisions, which significantly influence the final outcome.  As a result, the exact QoE for each user cannot be determined until the episode ends.
To address the challenge of sparse rewards in the learning process, we adopt a reward backpropagation mechanism to accurately assign rewards at each step.} As illustrated in Fig.~\ref{fig: DQN}, we first calculate the cumulative reward function--i.e., the objective function value--after the episode terminates {\color{blue}based on all  $o_i$ in the final state}. Then, each user’s achieved PAI at every step is  treated as an immediate reward, while the difference between the cumulative reward and the cumulative PAI from previous steps serves as the reward for the final step of the episode. The specific formulation is shown as
\begin{equation}
    {r_t} = \left\{ {\begin{array}{*{20}{c}}
{{\alpha _t}{\mathscr{F}}\left( {n_t^*} \right),}&{{\rm{done  =  FALSE,}}}\\
{\sum\nolimits_{i = t}^I {{\alpha _{i}}\mathscr{F}\left( {n_{i}^*} \right) - \sum\nolimits_{i = 1}^I {{L_i}} ,} }&{{\rm{done  =  TRUE,}}}
\end{array}} \right. \label{eq:reward}
\end{equation}
where $I$ represents the total number of users sending the request for offloading, and the flag ``done" serves as an indicator of whether the episode terminates. {\color{blue}Notably, as shown in the yellow block in Fig.~\ref{fig:flowchart}, ${\mathscr{F}}\left( {n_t^*} \right)$ is derived based on the convex optimization discussed in Section~\ref{sec:inner op}.  }

\subsubsection{PER-DQN-based policy design}
\label{sec: DQN}
Recall \eqref{eq:reward}; the objective function in \eqref{P1} is transformed into the expected cumulative reward ${\mathbb{E}}\left[ {\sum\nolimits_{t \in {\cal I}} {{\gamma ^t}{r_t}} } \right]$, with $\gamma = 1$.  Considering the discrete and small action space, we choose the simple, low-complexity DQN algorithm~\cite{mnih2015human} to learn a policy $\pi$ to maximize the action-state value at each step, which is defined as $Q\left( {{{\mathbf{s}}_t},{a_t}} \right) \triangleq \mathbb{E}\left[ {\sum\nolimits_{t' = t}^{{I_{{\text{end}}}}} {{\gamma ^{t'}}{r_{t'}}\left| {{{\mathbf{s}}_t},{a_t}} \right.} } \right]$. The optimal action-value function is defined as ${Q^*}\left( {{{\mathbf{s}}_t},a} \right) \triangleq {\max _\pi }{Q_\pi }\left( {{{\mathbf{s}}_t},{a_t}} \right)$, rom which the optimal policy can be  decided as ${\pi ^*} = \arg {\max _a}{Q^*}\left( {{{\mathbf{s}}_t},a} \right)$. Then, the Bellman equation  for ${Q^*}\left(  \cdot  \right)$ can be expressed by 
\begin{equation}
    {Q^*}\left( {{{\mathbf{s}}_t},{a_t}} \right) = \mathbb{E}\left[ {{r_t} + \gamma \mathop {\max }\limits_a {Q^*}\left( {{{\mathbf{s}}_{t + 1}},a} \right)\left| {{{\mathbf{s}}_t},{a_t},{\pi ^*}} \right.} \right].
\end{equation}
To enable the agent to accurately learn the Q-values corresponding to different $\left( {{{\mathbf{s}}_t},{a_t}} \right)$, and  more easily capture changes in the indicator, we use a neural network (NN) to model the relationship. Specifically, as shown in Fig.~\ref{fig: DQN},  we first embed the indicator within each local sub-state and then concatenate it with the static elements. Finally, the encoded state is flattened and fed into the dense layers.

\begin{figure*}
    \centering
    \includegraphics[width=0.9\linewidth]{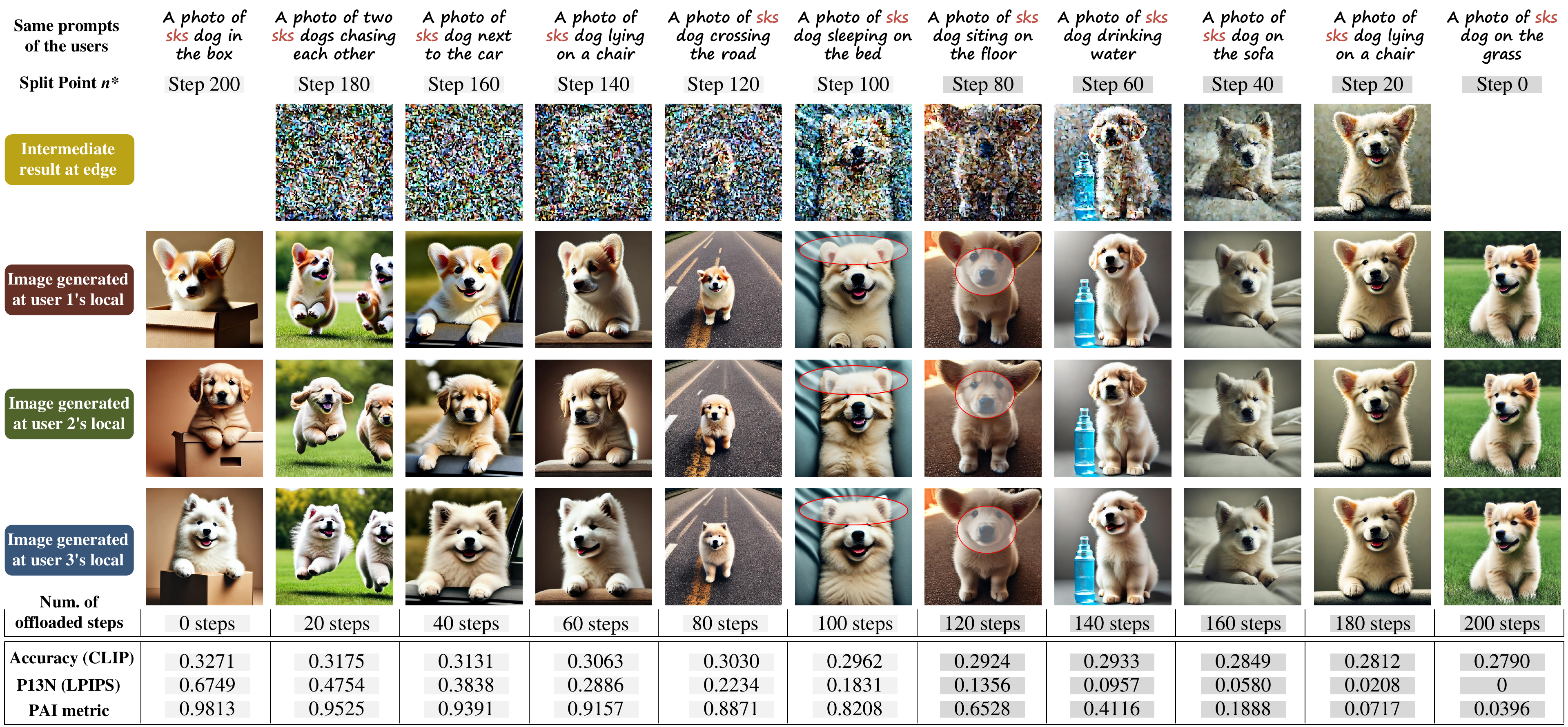}
    \caption{\color{blue}The visual results of hybrid inference. Therein, ``sks" in the prompts is the identifier for the personalized objects to the generated. Meanwhile, in evaluating the semantic accuracy using CLIP, we replace the original prompt term "sks" with personalized descriptive words such as "Corgi," "Golden Retriever," and "Samoyed",respectively for the three users. }
    \label{fig:virtualresult}
\end{figure*}

To help the agent find the optimal policy, we utilize a hybrid of \textit{$\varepsilon$-greedy exploration}~\cite{mnih2015human} and \textit{Boltzmann exploration}~\cite{cesa2017boltzmann}. Moreover, in order to keep the learning stability, we use two NNs with the same structure and the initial weights. One named current network is used to choose the action, the weights of which are denoted by ${{\boldsymbol{\theta }}_c}$ and update in each training step. The other named target network is used to calculate the Q-value, the weights of which are denoted by ${{\boldsymbol{\theta }}_t}$ and update periodically according to ${{\boldsymbol{\theta }}_c}$.
In each step, the agent samples a mini-batch of the experiences  from the \textit{memory} to train the NN. The loss function used in training is given by
\begin{equation}
{\mathcal{L}_{{\text{loss}}}} = \mathbb{E}\left[ {{{\left( {{y_t} - \left( {\mathop {\max }\limits_a {Q^*}\left( {{{\mathbf{s}}_t},a;{{\mathbf{\boldsymbol \theta }}_c}} \right)} \right)} \right)}^2}} \right],
\end{equation}
where ${y_{t = }}{r_t} + \gamma \mathop {\max }\limits_a {Q^*}\left( {{{\mathbf{s}}_{t + 1}},a;{{\mathbf{\boldsymbol \theta }}_t}} \right)$.
Moreover, given the high variance in experiences, we assign the experience about the terminal state more attention. When sampling experiences, we set the ratio of selected terminal states to preceding states at 1:7. Meanwhile, we also employ Prioritized Experience Replay (PER)~\cite{schaul2015prioritized}, where the  experiences are sampled based on their significance, prioritizing transitions with higher learning potential to improve training efficiency. The setting of hyperparameters is detailed in Section~\ref{sec:simulation}.

\begin{figure}
    \centering    \includegraphics[width=0.9\linewidth]{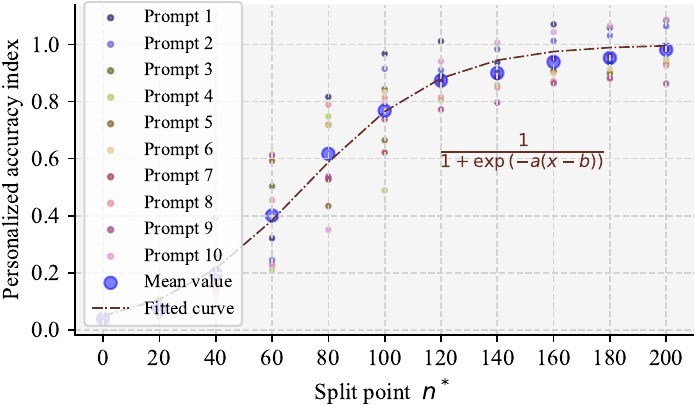}
    \caption{The fitting curve for the relation between PAI and  split point, along with the mean values of
accuracy and personalization for CLIP and LPIPS metrics across 200 sets of image generated from 10  prompts at
each split points.}
    \label{fig:fittingcurve}
    \vspace{-0.3cm}
\end{figure}
\section{Simulation and Evaluation}
\label{sec:simulation}
\textcolor{blue}{In this section, we detail the fine-tuning of the cluster-wide and local personalized models, along with PAI function fitting. We then describe the hybrid inference simulation setup and analyze PER-DQN performance. All related code, datasets, and fine-tuned models are available at \url{https://github.com/wty2011jl/E-MOPDM}}.

\subsection{Pre-deployment Fine-tuning \& PAI Fitting}
\textcolor{blue}{In this work, we employ the Dreambooth technique \cite{ruiz2023dreambooth} to fine-tune a cluster-wide model and three local personalized models. Each local dataset includes seven photos of the corresponding user's pet. The cluster-wide model's dataset comprises nine images: six selected from local datasets and three generated using prompts with detailed personalized descriptions. We fine-tune the text encoder and UNet using an NVIDIA GeForce RTX 4090, to redefine the specific representation corresponding to ``dog" for each user individually. Using the fine-tuned models, we simulate the hybrid inference process with the total denoising steps of 200 using 10 prompts, each prompt generate 20 images for each users. The virtual results have been shown in Fig.~\ref{fig:virtualresult}. To facilitate comparison, we present the personalized results generated by three users for each prompt. However, in our algorithm design, we assume that the prompts for each user are distinct.
From the figure,  we observe that  when  $n^* > 80$, shared offloading leads to uniform image layouts across three users, yet the dog portraits remain successfully personalized locally without additional detailed descriptions. Conversely, when $n^* < 80$, there is a significant decrease in personalized accuracy. Specifically, only the facial expressions retain the distinct characteristics of different dog breeds, while highly differentiated features such as the ears are not successfully altered.  Therefore, we choose $\hat N = 80$.
In addition, we can also observe that as the number of offloading steps increases, only the personalized aspects deteriorate, and people's visual evaluations align with the changes in our PAI metric defined in \eqref{eq:PAI}. Moreover,  we present a scatterplot showing the mean PAI values for 20 images generated per prompt across 10 different prompts, with varying numbers of the split point $n^*$
  as depicted in Fig.~\ref{fig:fittingcurve}. We then apply a sigmoid fitting to these scatter points, with parameters 
 $a_{\mathscr{F}} = 0.0413$ and $b_{\mathscr{F}} = 71.44$. }
\subsection{Hybrid Inference Simulation Setup}
\subsubsection{Scenario parameter settings}
We consider a scenario involving a single edge server and multiple users sending offloading requests within the past  $\Delta$. In the following simulations, the determination of inference latency and PAI is based on empirical functions derived from actual experiments, as shown in Figs.~\ref{fig:gpu} and~\ref{fig:fittingcurve}. Specifically, we assume that the edge server is equipped with $G$ GPUs, all of the H100 NVL, while users have a single, varying local GPU, which could be one of the other types: $\left\{ {{\text{RTX 2060, GTX 1080,}} \ldots {\text,~}{\text{RTX }}4090} \right\}$. Meanwhile, the emphasis $\alpha_i$ of user $i$ on the PAI and the latency is randomly selected from the interval $\left[ {\frac{{\left( {{\mathscr{L}_i}\left( 1 \right) - {\mathscr{L}_e}\left( {{{\hat b}_{\text{g}}},G} \right)} \right)}}{{{a_\mathscr{F}}\mathscr{F}(\hat N)(1 - \mathscr{F}(\hat N))}},\frac{{\kappa  \left( {{\mathscr{L}_i}\left( 1 \right) - {\mathscr{L}_e}\left( {{{\hat b}_{\text{g}}},G} \right)} \right)}}{{{a_\mathscr{F}}\mathscr{F}(N)(1 - \mathscr{F}(N))}}} \right]$, where we set ${{{\hat b}_{\text{g}}}} = 20$ and $\kappa  = 0.05$. The introduction of $\kappa$ accounts for the fact that users sending offloading requests tend to be more sensitive to latency to varying extend. Moreover, we set the total bandwidth $W_{\max} = 1$ MHz. Assuming the base station employs power control, the spectral efficiency for each user is set as $\eta = 10$ bits/s/Hz. Considering that the data sizes of the intermediate noisy image and prompt of different users are approximately equal, we substitute the specific noisy image and prompt for each user with their respective means. Thus, we have $s_i^{\rm p} = 216$ b, and $s_i^{\rm m} = 4.4$ Mb.

\subsubsection{PER-DQN hyperparameter settings}
The Q-network consists of one embedding layer,  three hidden linear layers and an output layer. The vocabulary size and embedding dimension of the embedding layer are set as 4 and 3, respectively. Each of the hidden  linear layer has 256 nodes. The output layer has two nodes. The target network is updated at a frequency of once every
2000 time steps. The memory capacity is 400000, the batch size is 128, and the learning rate is 0.0001.  The exploration probability $\varepsilon$  linearly decreases from 0.5 to 0.001. The temperature parameter of Boltzmann exploration linearly decreases from 5 to 0.01. During PER, the priority exponent is set to 0.7, the importance sampling weight is set to 0.3, and the priority offset is set to 0.00002. Moreover, during the training, the reward is scaled to 0.1 of its original value. {\color{blue}Moreover, the training is performed on a  GeForce RTX 4070 Ti. To save training time, the inference latency of different GPUs is obtained from prior empirical measurements on various GPU models, as summarized in Fig.~\ref{fig:gpu}}.
\begin{figure*}[ht]
    \centering
    \subfloat[\label{fig:offload_a}]{%
        \includegraphics[width=0.4\textwidth]{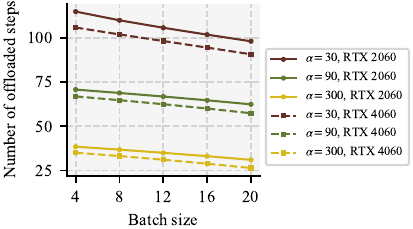} 
    }
    \qquad \qquad
    \subfloat[\label{fig:offload_b}]{%
        \includegraphics[width=0.4\textwidth]{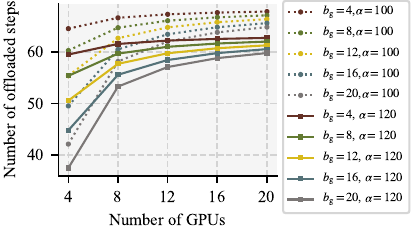} 
    }
    \caption{The relationship between the number of denoising steps offloaded and user emphasis $\alpha_i$ under varying configurations. (a) Analysis with differing batch sizes at the edge server and varying computational capacities of the local terminal. (b) Impact of different numbers of available GPUs on the edge server combined with varying batch sizes.}
    \label{fig:num_offloaded step}
\end{figure*}
\begin{figure*}[ht]
    \centering
    \subfloat[\label{fig:learning_1}]{%
        \includegraphics[width=0.29\textwidth]{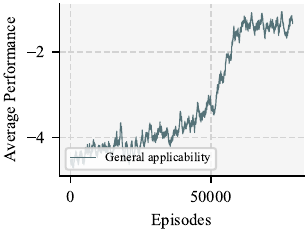} 
    }
    \hfill
    \subfloat[\label{fig:learning_2}]{%
        \includegraphics[width=0.29\textwidth]{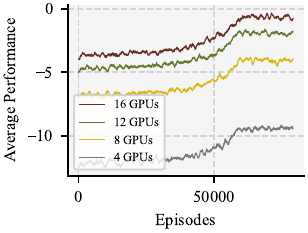} 
    }
    \hfill\subfloat[\label{learning_3}]{%
        \includegraphics[width=0.29\textwidth]{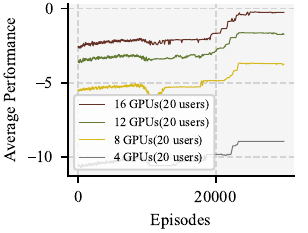} 
    }
    \caption{
Learning curves for DRL algorithms across varying scopes: (a) General application with 2000 random seeds and smoothing over 1000 points. (b) GPU-constrained application with 1000 random seeds and similar smoothing. (c) Specific application with a single seed and smoothing over 100 points.}
    \label{fig:learning_curve}
\end{figure*}
\begin{figure*}[ht]
    \centering
    \subfloat[\label{fig:application1}]{%
        \includegraphics[width=0.33\textwidth]{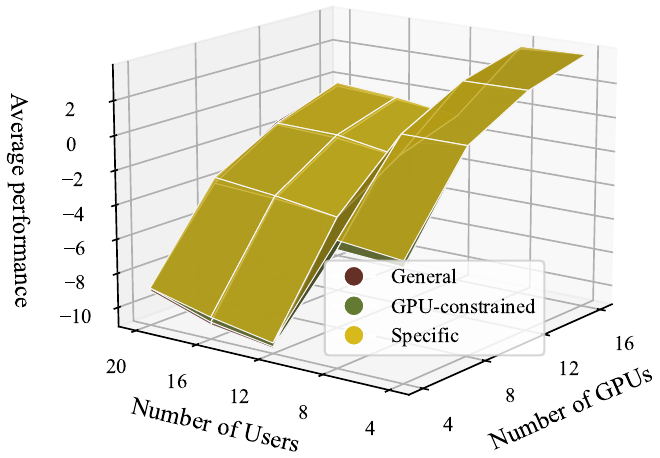} 
    }
    \hfill
    \subfloat[\label{fig:application2}]{%
        \includegraphics[width=0.31\textwidth]{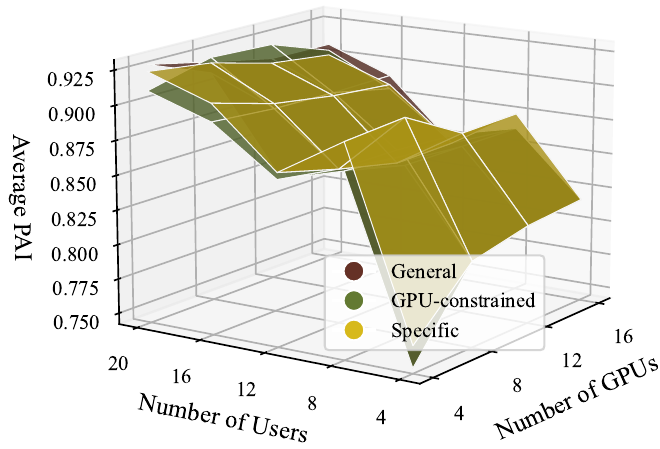} 
    }
    \hfill\subfloat[\label{fig:application3}]{%
        \includegraphics[width=0.33\textwidth]{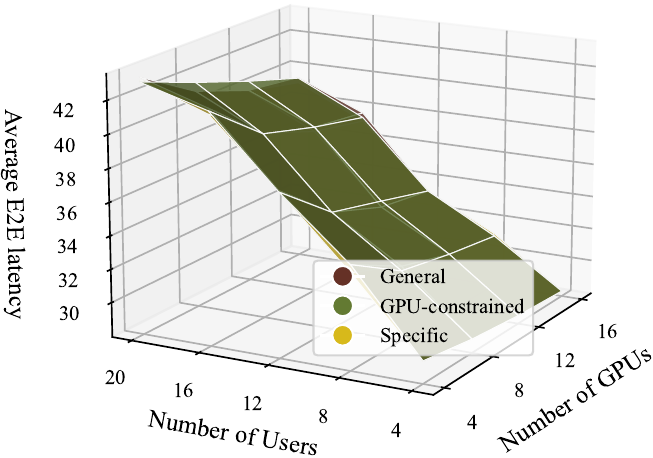} 
    }
    \caption{Performance comparison of three DRL algorithms with different applicability scopes for the same specific case. (a) Overall all performance. (b) Average PAI. (c) Average E2E latency.}
    \label{fig:application}
\end{figure*}
\subsection{Hybrid Inference Performance Analysis}
Firstly, we analyze the key factors influencing the optimal offloaded denoising steps and their respective impacts. As shown in Fig.~\ref{fig:offload_a}, given an edge server with 8 available GPUs, as the batch size at the edge increases, the inference latency per denoising step at the edge server lengthens, and the optimal number of offloaded steps gradually decreases. Additionally, for the same batch size, the stronger the computational capacity of the user’s local device, the fewer offloaded steps are required, with more denoising steps executed locally.
Furthermore, as illustrated in Fig.~\ref{fig:offload_b}, with an increase in the number of GPUs on the edge server and the resulting improvement in inference speed, the optimal number of offloaded steps also increases. Both sub-figures indicate that as the value of $\alpha_i$ increases, reflecting heightened user focus on PAI, the optimal offloaded steps decrease accordingly. Moreover, when computational resources are relatively abundant, 
$\alpha_i$ exerts a dominant influence on the optimal offloading steps. Conversely, with a smaller number of available GPUs, batch size has a more significant impact on the optimal number of offloaded steps.

\begin{figure*}[ht]
    \centering
    \subfloat[\label{fig:cc_algorithm_1}]{%
        \includegraphics[width=0.33\textwidth]{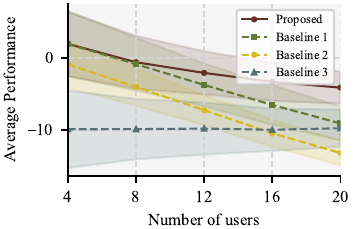} 
    }
    \subfloat[\label{fig:cc_algorithm_2}]{%
        \includegraphics[width=0.33\textwidth]{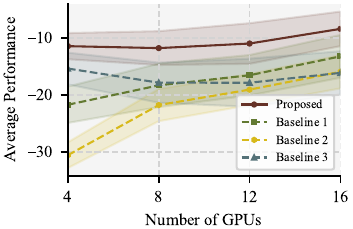} 
    }
    \subfloat[\label{fig:cc_algorithm_3}]{%
        \includegraphics[width=0.33\textwidth]{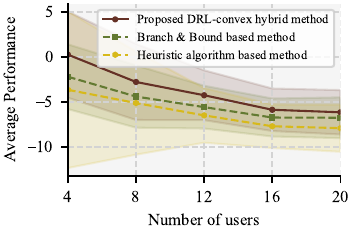} 
    }
    \caption{Performance comparison : (a) Comparison between three baselines: Average performance vs. number of users (100 cases/user count); (b) Comparison between three baselines: Average performance vs. number of GPUs on edge server (100 cases/GPU configuration). (c)Comparison between DRL-convex hybrid, Branch \& Bound, and Heuristic methods (100 cases/user count).}
    \label{fig:cc_algorithm}
\end{figure*}

Next, we integrate the aforementioned split point optimization into the modeling of environmental transitions in the MDP model. Through simulations, we verify the effectiveness of the proposed DQN-convex hybrid solution. Given different applicability scopes, we train three types of Q-networks tailored for the following three scenarios:
\begin{itemize}
    \item \textit{General Applicability}: For scenarios with evolving computational resources at the edge, dynamic numbers of users, and various user combinations, with a greater emphasis on the algorithm's generalizability.
    \item \textit{GPU-Constrained Applicability}: For cases with specific computational resources, dynamic numbers of users, and various user combinations, striking a balance between generalizability and optimality.
    \item \textit{Specific Applicability}: For scenarios with fixed computational resources, a specific number of users, and a particular user combination, focusing more the with a greater focus on the algorithm's optimality.
\end{itemize}
The learning curves for the three application scopes are shown in Fig.~\ref{fig:learning_curve}. From the figure, we can observe that the agent can learn effectively in all three cases. Additionally, the more generalizable the model, the greater the fluctuation in its convergence curve, which is mainly due to the inherent diversity of the environment. Moreover, the performance of models with different applicability scopes is shown in Fig.~\ref{fig:application}. From Fig.~\ref{fig:application1},  we observe that the model tailored to specific applications achieves  generally superior performance. Specifically, excluding the impact of local computational heterogeneity and user preferences, the average performance tends to increase as the number of users decreases and the number of available GPUs on the edge server increase. As shown in Fig.~\ref{fig:application2} and~\ref{fig:application3}, this improvement is primarily due to the inference latency reduction achieved through offloading, which outweighs the performance loss from potential PAI degradation. {\color{blue}However, it is important to note that each performance value in this figure is derived from a specific configuration. Thus, the increase in user count from 8 to 12 not only represents a 1/3 increase in users, but the additional users' computational power and emphasis parameters also impact performance. As shown in Fig.~\ref{fig:cc_algorithm_1}, when 100 cases are randomly generated for different user counts and then averaged, the sudden "jump" from 12 to 8 in the curve is smoothed out.}
Given the slight performance difference across the three algorithms, the model for general application is preferred in highly dynamic environments, as it avoids the computational cost associated with frequent retraining and fine-tuning. For scenarios with a fixed set of requesting users, the model tailored for specific applications can be selected to achieve optimal performance. The model for GPU-constrained applications strikes a balance between generality and optimality by training separate models tailored to each GPU availability scenario, leveraging the model redundancy to accommodate varying GPU constraints.

In the subsequent algorithm comparisons, we focus on the more common dynamic scenarios and use the \textit{general model} for performance analysis. Firstly, we compare this algorithm with the commonly used straightforward approaches for personalized diffusion model inference. Three baseline methods considered are as follows:
\begin{itemize}
    \item \textit{Baseline~1}: Within the maximum batch size supported by the edge server, all users offload to the edge, with optimization of the split point.
    \item \textit{Baseline~2}: Within the maximum batch size supported by the edge server, all users offload to the edge, with a fixed offloading of $\hat N$ steps.
    \item \textit{Baseline~3}: All users perform the complete inference locally. 
\end{itemize}
For {baseline~1} and {baseline~2}, if the total number of users sending requests exceeds the maximum batch size supported by the edge server, users who send requests later perform inference locally. Fig.~\ref{fig:cc_algorithm} presents a comparison of the proposed algorithm with three baseline methods, illustrating the performance trends as the number of users and the available GPUs on the edge server vary. Specifically, in Fig.~\ref{fig:cc_algorithm_1}, we fix the number of available GPUs on the edge server to 8, while in Fig.~\ref{fig:cc_algorithm_2}, we fix the number of users in each case to 20. Observing both subfigures, baseline~3, where all inference is completed locally by users, shows no noticeable trend in average performance as the number of users or available GPUs changes. In contrast, for the other three methods, average performance declines as the number of users increases or the number of GPUs decreases. Furthermore, since baseline~1 incorporates split point optimization compared to baseline~2, baseline~1 consistently maintains a stable performance gap above baseline~2. Additionally, due to its joint optimization of split point and request handling, the performance variation of the proposed method is narrower than that of baselines~1 and~2, and it consistently outperforms other baseline algorithms. This result further indicates that the algorithm can flexibly adaptively trade off between PAI and latency across different scenarios, achieving optimal overall performance.

Finally, we demonstrate the effectiveness and optimality of the DRL-convex hybrid solution for the extended GQAP by comparing it with two widely used methods for GQAP: a heuristic algorithm and a branch \& bound-based Integer Linear Programming (ILP) method. The heuristic algorithm selects is a genetic algorithm with a population size of 
$P=100$ and an iteration count of 
$M=200$. Additionally, as the optimal split point lacks a closed-form solution, we omit the split point optimization in the Branch \& Bound-based ILP method for simplicity. As illustrated in Fig.~\ref{fig:cc_algorithm_3}, the proposed DRL-convex hybrid solution outperforms the other methods, achieving superior performance. Specifically, due to the omit of the split point optimization, the Branch \& Bound based method, the optimality of the method consistently remains lower than that of our proposed method. Additionally, although the heuristic algorithm does not require the rigorous mathematical formulation demanded by the Branch \& Bound based method, it can incorporate split point optimization. However, its stability in terms of optimality is noticeably lower than that of the other two algorithms. As a result, while it may occasionally achieve the same level of optimality as our proposed DRL-convex hybrid method, its average performance is the poorest. Moreover, the complexity of the three methods with with respect to user
count $I$ as shown in Table~\ref{tbl:complexity}. From the table, since in the process of constructing the MDP model, we transform the original problem with high-dimensional optimization variables into a decision sequence problem focused on a single variable, the proposed DRL-convex hybrid method shows a linear complexity of $O\left(I\right)$. This linear growth suggests that the DRL-Convex Hybrid Method is scalable and well-suited for scenarios with a large number of users. In contrast, the Heuristic Algorithm-Based Method has a complexity of 
$O\left( {2 \cdot P \cdot M \cdot I} \right)$. Although it also scales linearly with $I$, its overall complexity is influenced by $P$ and 
$M$, requiring careful parameter tuning to ensure efficiency as the user count rises.
In addition,
the conventional Branch \& Bound-Based Method presents an exponential complexity of 
$O\left( {{2^{2 \times I}}} \right)$, which grows rapidly with increasing 
$I$. This exponential growth makes it 
 work for small-scale problems and impractical for large-scale scenarios, as the computational cost becomes prohibitive. Overall, with a well-trained Q network,
 the DRL-convex hybrid algorithm is more viable for larger user counts due to their manageable complexity.
\begin{table}[]
\footnotesize
\centering
\caption{Complexity comparison  regarding user count.}
\begin{tabular}{c|l|c}
\hline
\multicolumn{1}{c|}{Order$\uparrow$} & Name of Method                   & \multicolumn{1}{c}{Complexity} \\ \hline \hline
1                                     & DRL-convex hybrid         & $O\left( I \right)$                                        \\ \hline
2                                     & Heuristic algorithm  & $O\left( {2 \cdot P \cdot M \cdot I} \right)$              \\ \hline
3                                     & Branch \& Bound     & $O\left( {{2^{2 \times I}}} \right)$                       \\ \hline
\end{tabular}
\vspace{-0.5cm}
\label{tbl:complexity}
\end{table}

\section{Conclusion}
\label{sec:conclusion}
We proposed an efficient offloading framework for deploying personalized SDMs in multi-user scenarios with diverse computing capabilities. 
To balance latency and accuracy, we have introduced a tunable emphasis parameter and formulated offloading and split-point optimization as an extended GQAP. A DRL-convex hybrid approach has enabled real-time decision-making. Simulations have demonstrated the effectiveness of the proposed framework and solutions.

{\color{blue}However, several limitations remain: 1) In this work, the cluster-wide model and local models share the same size and are independently trained on different datasets. This may limit the cluster-wide model’s ability to capture common features effectively, and the impact of switching between inference models has not been considered. 2) Our approach assumes a given cluster as the starting point. However, key aspects such as evaluating user similarity within a cluster, defining cluster partitions, and understanding the relationship between cluster-wide model generalization and individual users’ final PAI remain unexplored.}

In light of above, in future work, we will pursue further optimization from the three perspectives: 1) to enhance the generalization capability of the model, we will design training optimization methods for cluster-specific models, achieving a trade-off between model size and {\color{blue}PAI} performance; 2) we will explore optimal clustering methods based on task similarity to balance the storage and computational energy consumption of multi-cluster-wide models with {\color{blue}PAI} performance.  {\color{blue}3) We will develop a joint training algorithm for both the cluster-wide model and the personalized local models based on the federated learning with the ingenious hierarchical clustering-based aggregation method. This can also address the scenarios where users have insufficient local data.}
\bibliographystyle{IEEEtran}
\bibliography{ref}

\begin{IEEEbiography}[{\includegraphics[width=1in,height=1.25in,clip,keepaspectratio]{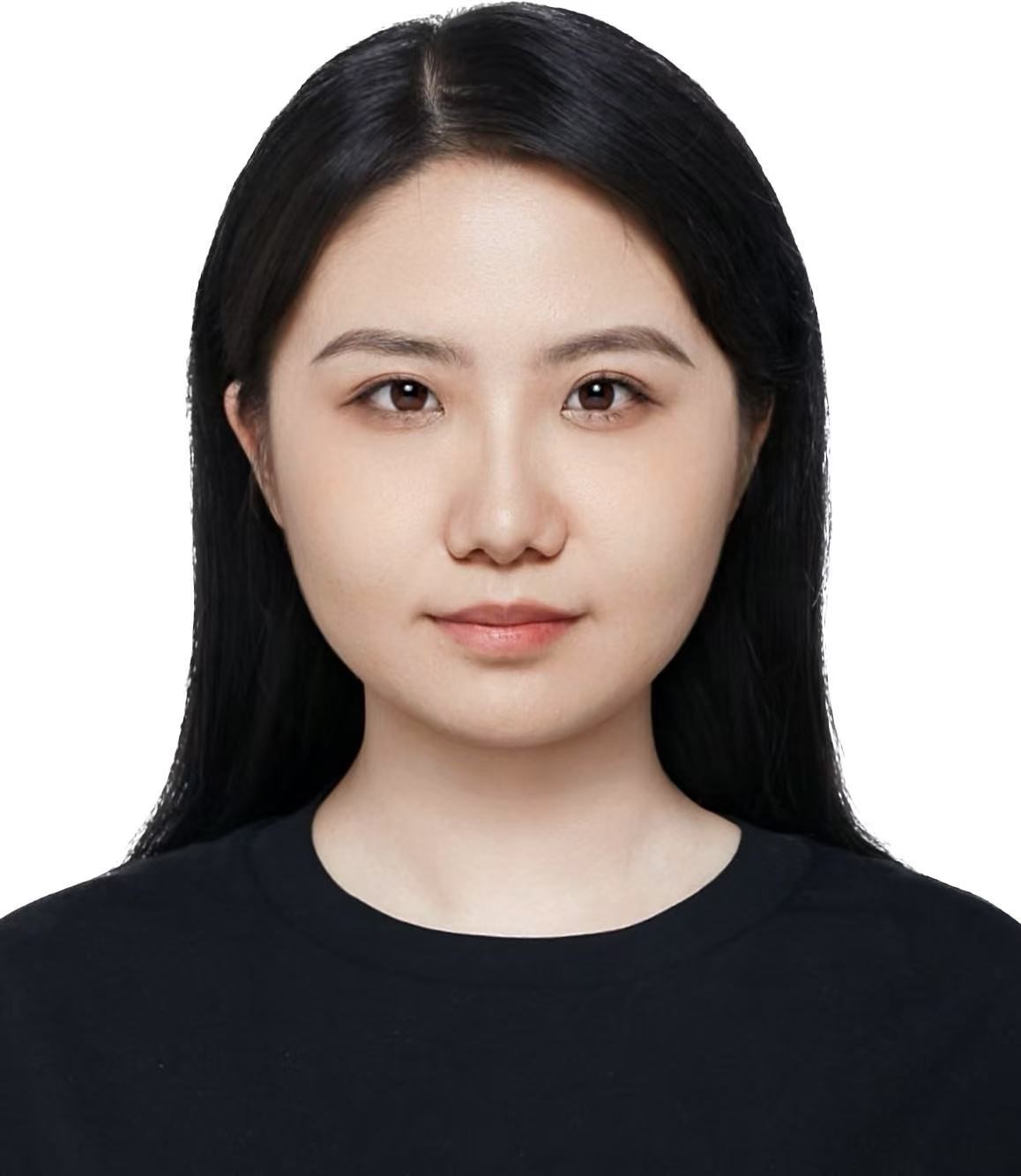}}]{Wanting Yang}

is currently a research fellow and scientist with Singapore University of Technology and Design. She received the B.S. degree and the Ph.D. degree from the Department of Communications Engineering, Jilin University, Changchun, China, in 2018 and 2023, respectively. She was a visiting student with Singapore University of Technology and Design. She
served as Technical Programme Committee member and reviewers in flagship
conferences, such as WCNC, Globecom, ICC, and VTC, and top journals. Her research interests include  semantic communication, deep reinforcement learning, martingale theory, edge computing, edge intelligence, generative AI.
\end{IEEEbiography}

\begin{IEEEbiography}[{\includegraphics[width=1in,height=1.25in,clip,keepaspectratio]{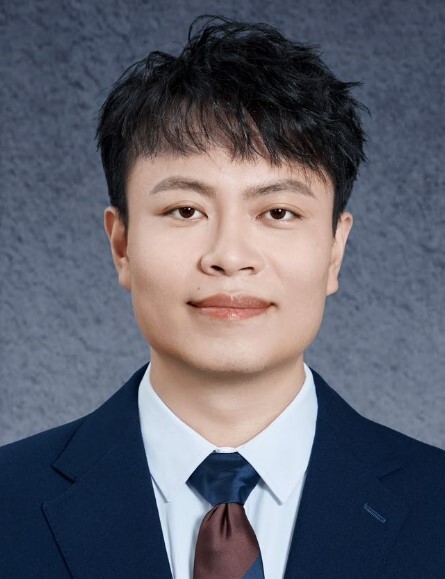}}]{Zehui Xiong} is currently with Singapore University of Technology and Design, Singapore.  Recognized as a Highly Cited Researcher, he has published more than 250 research papers in leading journals, and he has won over 10 Best Paper Awards in international conferences.  He is the recipient of Forbes Asia 30u30, IEEE Asia Pacific Outstanding Young Researcher Award, IEEE Early Career Researcher Award for Excellence in Scalable Computing, IEEE Technical Committee on Blockchain and Distributed Ledger Technologies Early Career Award, IEEE Internet Technical Committee Early Achievement Award, IEEE TCSVC Rising Star Award, IEEE TCI Rising Star Award, IEEE TCCLD Rising Star Award, IEEE ComSoc Outstanding Paper Award, IEEE Best Land Transportation Paper Award, IEEE Asia Pacific Outstanding Paper Award, IEEE CSIM Technical Committee Best Journal Paper Award, IEEE SPCC Technical Committee Best Paper Award, IEEE Big Data Technical Committee Best Influential Conference Paper Award, and IEEE VTS Singapore Best Paper Award. He has served as the Associate Director of Future Communications R\&D Programme, and Deputy Lead of AI Mega Centre.
\end{IEEEbiography}

\begin{IEEEbiography}[{\includegraphics[width=1in,height=1.25in,clip,keepaspectratio]{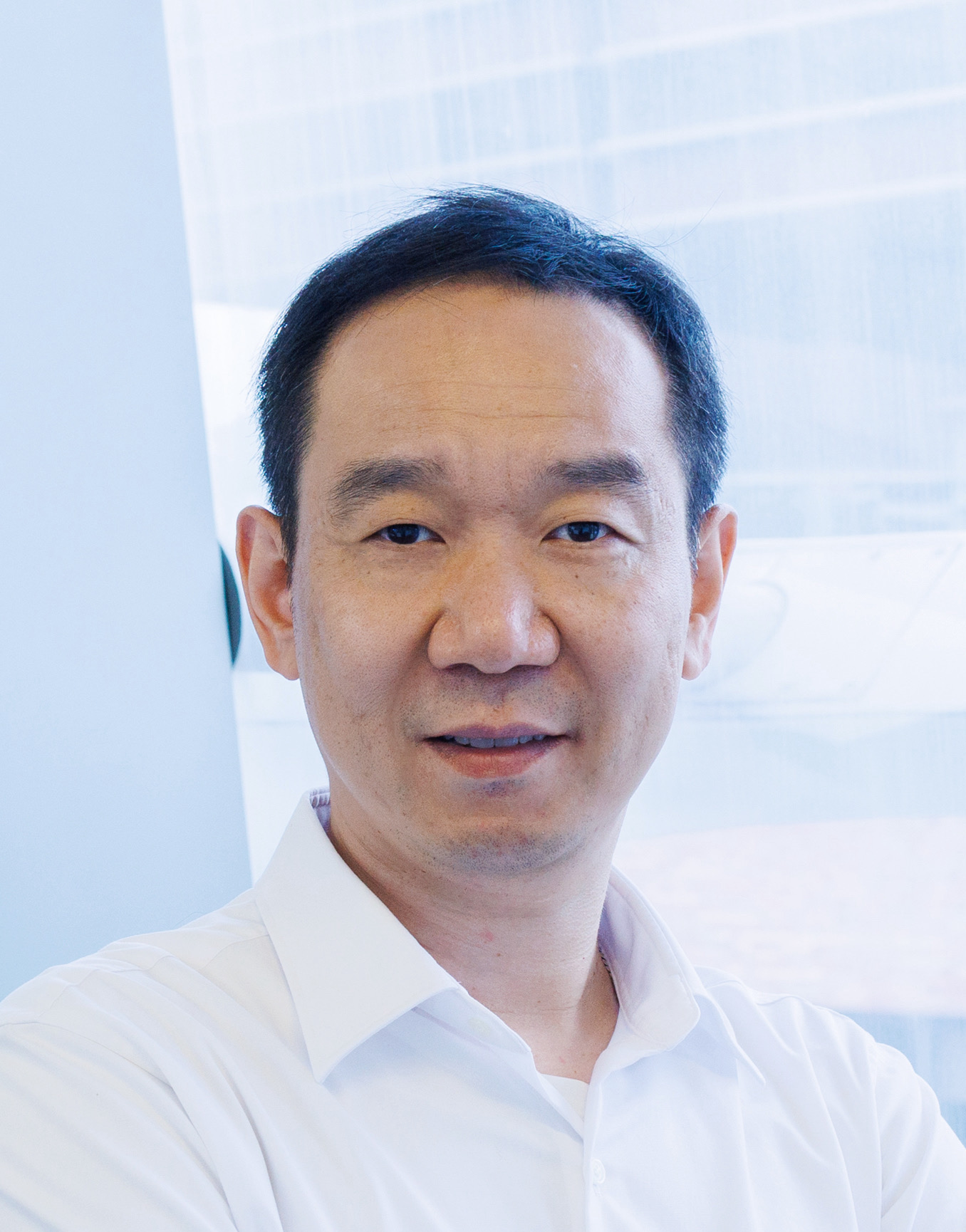}}]{Song Guo}
is a Full Professor in the Department of Computer Science and Engineering at the Hong Kong University of Science and Technology. He also holds a Changjiang Chair Professorship awarded by the Ministry of Education of China. His research interests are mainly in the areas of big data, edge AI, mobile computing, and distributed systems. With many impactful papers published in top venues in these areas, he has been recognized as a Highly Cited Researcher (Web of Science) and received over 12 Best Paper Awards from IEEE/ACM conferences, journals, and technical committees. Prof. Guo is the Editor-in Chief of IEEE Open Journal of the Computer Society. He has served on the IEEE Communications Society Board of Governors, IEEE Computer Society Fellow Evaluation Committee, and editorial board of a number of prestigious international journals like IEEE Transactions on Parallel
and Distributed Systems, IEEE Transactions on Cloud Computing, IEEE Internet of Things Journal, etc. He has also served as chair of organizing and technical committees of many international conferences. Prof. Guo is an IEEE Fellow and an ACM Distinguished Member.
\end{IEEEbiography}
\begin{IEEEbiography}[{\includegraphics[width=1in,height=1.25in,clip,keepaspectratio]{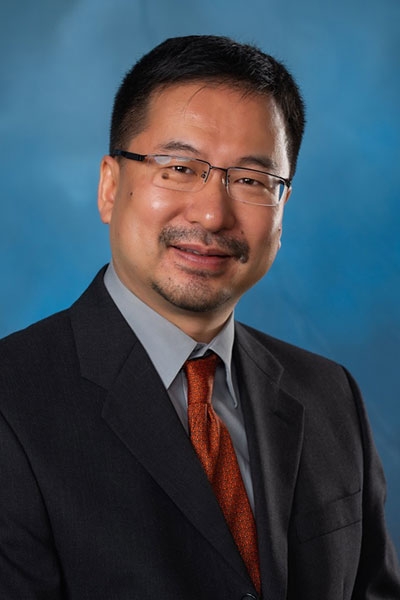}}]{Shiwen Mao}
 is a Professor and the Earle C. Williams Eminent Scholar Chair, and the Director of the Wireless Engineering Research and Education Center, Auburn University, Auburn, AL, USA. His research interest includes wireless networks, multimedia communications, and smart grid. He received the IEEE ComSoc MMTC Outstanding Researcher Award in 2023, the IEEE ComSoc TC-CSR Distinguished Technical Achievement Award in 2019, and the NSF CAREER Award in 2010. He is a co-recipient of the 2022 Best Journal Paper Award of IEEE ComSoc eHealth Technical Committee, the 2021 Best Paper Award of Elsevier/Digital Communications and Networks (KeAi), the 2021 IEEE Internet of Things Journal Best Paper Award, the 2021 IEEE Communications Society Outstanding Paper Award, the IEEE Vehicular Technology Society 2020 Jack Neubauer Memorial Award, the 2018 ComSoc MMTC Best Journal Paper Award and the 2017 Best Conference Paper Award, the 2004 IEEE Communications Society Leonard G. Abraham Prize in the Field of Communications Systems, and several ComSoc technical committee and conference best paper/demo awards. He is the Editor-in-Chief of IEEE TRANSACTIONS ON COGNITIVE COMMUNICATIONS AND NETWORKING. He is a Distinguished Lecturer of IEEE Communications Society and the IEEE Council of RFID.
\end{IEEEbiography}
\begin{IEEEbiography}[{\includegraphics[width=1in,height=1.25in,clip,keepaspectratio]{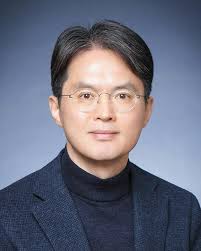}}]{Dong In Kim} received the Ph.D. degree in electrical engineering from the University of Southern California, Los Angeles, CA, USA, in 1990. He was a Tenured Professor with the School of Engineering Science, Simon Fraser University, Burnaby, BC, Canada. He is currently a Distinguished Professor
with the College of Information and Communication Engineering, Sungkyunkwan University, Suwon, South Korea. He is a Fellow of the Korean Academy of Science and Technology and a Life Member of the National Academy of Engineering of Korea. He was the first recipient of the NRF of Korea Engineering Research Center (ERC) in Wireless Communications for RF Energy Harvesting from 2014 to 2021. He received several research awards, including the 2023 IEEE ComSoc Best Survey Paper Award and the 2022 IEEE Best Land Transportation Paper Award. He was selected the 2019 recipient of the IEEE ComSoc Joseph LoCicero Award for Exemplary Service to Publications. He was
the General Chair of the IEEE ICC 2022, Seoul. From 2001 to 2024, he served as an Editor, an Editor at Large, and an Area Editor of Wireless Communications I for IEEE Transactions on Communications.  He has been listed as a 2020/2022 Highly Cited Researcher by Clarivate Analytics.
\end{IEEEbiography}

\begin{IEEEbiography}[{\includegraphics[width=1in,height=1.25in,clip,keepaspectratio]{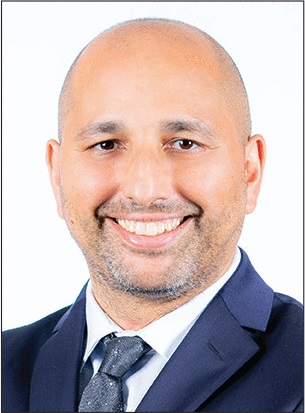}}]{Mérouane Debbah}
is a Professor at Khalifa University of Science and Technology in Abu Dhabi and founding Director of the KU 6G Research Center. He is a frequent keynote speaker at international events in the field of telecommunication and AI. His research has been lying at the interface of fundamental mathematics, algorithms, statistics, information and communication sciences with a special focus on random matrix theory and learning algorithms. In the Communication field, he has been at the heart of the development of small cells (4G), Massive MIMO (5G) and Large Intelligent Surfaces (6G) technologies. In the AI field, he is known for his work on Large Language Models, distributed AI systems for networks and semantic communications. He received multiple prestigious distinctions, prizes and best paper awards for his contributions to both fields. He is a WWRF Fellow, a Eurasip Fellow, an AAIA Fellow, an Institut Louis Bachelier Fellow and a Membre émérite SEE.
\end{IEEEbiography}
\end{document}